\pdfoutput=1
\documentclass[amsfonts, amsmath, amssymb, aps, prd, reprint]{revtex4-2} 
\usepackage{amsthm} 
    \newtheorem{theorem}{Theorem}
    \newtheorem{corollary}{Corollary}
    \newtheorem{proposition}{Proposition}
    
    \theoremstyle{definition}
    \newtheorem{definition}{Definition}
\usepackage{graphicx} 
\usepackage{array} 
    \newcounter{rowcntr}[table]
    \newcolumntype{N}{>{\refstepcounter{rowcntr}\alph{rowcntr}}c} 
    \AtBeginEnvironment{tabular}{\setcounter{rowcntr}{0}} 
\usepackage[utf8]{inputenc}
\usepackage[T1]{fontenc} 
\usepackage{microtype} 
\usepackage{physics} 
    \providecommand{\vcentcolon}{\mathrel{\mathop{:}}}
    \newcommand{\normord}[1]{\vcentcolon\mathrel{#1}\vcentcolon} 
\usepackage{tensor} 
\usepackage{siunitx} 
    \DeclareSIUnit{\year}{\text{yr}}
\usepackage[svgnames]{xcolor} 
    \definecolor{NickRed}{HTML}{e20134} 
    \definecolor{NickOrange}{HTML}{f1611a}
    \definecolor{NickYellow}{HTML}{ffc100}
    \definecolor{NickGreen}{HTML}{04c421}
    \definecolor{NickBlue}{HTML}{0062ff}
    \definecolor{NickPurple}{HTML}{6610f2}
\usepackage{pgfplots} 
    \pgfplotsset{compat=1.18}
    \usepgfplotslibrary{groupplots} 
    \usepgflibrary{plotmarks} 
    \usetikzlibrary{arrows.meta,calc} 
\usepackage{fontawesome5} 
\usepackage{pgfplotstable} 

\pgfplotstableread{data2.dat}\tabletwo%
\pgfplotstablecreatecol[create col/expr={-log10(-\thisrow{M})}]{log}\tabletwo%

\pgfplotstableread{data3.dat}\tablethree%
\pgfplotstablecreatecol[create col/expr={-log10(-\thisrow{M})}]{log}\tablethree%

\pgfplotstableread{data4.dat}\tablefour%
\pgfplotstablecreatecol[create col/expr={-log10(-\thisrow{M})}]{log}\tablefour%

\pgfplotstableread{data5.dat}\tablefive%
\pgfplotstablecreatecol[create col/expr={-log10(-\thisrow{M})}]{log}\tablefive%

\usepackage[allcolors=blue,colorlinks,pdftitle={Positive Mass in General Relativity Without Energy Conditions},pdfauthor={Níckolas de Aguiar Alves, André G. S. Landulfo, Bruno Arderucio Costa}]{hyperref} 
\usepackage{ORCIDinREVTeX} 

\newcommand{\ie}{\emph{i.e.}} 
\newcommand{\eg}{\emph{e.g.}} 
\begin{document}

\title{Positive Mass in General Relativity Without Energy Conditions}

\author{Níckolas de \surname{Aguiar Alves}}
\orcid{0000-0002-0309-735X}
\email{alves.nickolas@ufabc.edu.br}    
\affiliation{Center for Natural and Human Sciences, Federal University of ABC,\\Av. dos Estados, 5001, 09210-580, Santo André, SP, Brazil}

\author{André G. S. Landulfo}
\orcid{0000-0002-3717-4966}
\email{andre.landulfo@ufabc.edu.br}
\affiliation{Center for Natural and Human Sciences, Federal University of ABC,\\Av. dos Estados, 5001, 09210-580, Santo André, SP, Brazil}

\author{Bruno Arderucio Costa}
\orcid{0000-0001-5182-2010}
\email{bcosta@troy.edu}
\affiliation{Center for Relativity and Cosmology, Troy University,\\ Troy, AL, 36082, USA}

\date{January 17, 2025}

\begin{abstract}
A long-standing problem in physics is why observed masses are always positive. While energy conditions in quantum field theory can partly answer this problem, in this paper we find evidence that classical general relativity abhors negative masses, without the need for quantum theory or energy conditions. This is done by considering many different models of negative-mass ``stars'' and showing they are dynamically unstable. \emph{A fortiori}, we show that any barotropic negative-mass star must be dynamically unstable.
\end{abstract}

\maketitle

\section{Introduction}\label{sec: intro}
A simple, yet profound, question one can ask about gravity is: ``why do things fall down?'' This is a conspicuous experimental fact about gravitational physics, and it is certainly fundamental for the formation of structures in the Universe and the existence of life itself. The answer, however, may dive deep into the foundations of physics. 

From a purely quantum field theoretical point of view, one could answer that things fall down because gravity is mediated by a spin \num{2} field---the graviton. As is well-known \cite{percacci2017IntroductionCovariantQuantum,peskin1995IntroductionQuantumField}, even-spin mediators lead to attractive forces for particles with charges with the same sign. Hence, positive masses must attract each other. This is in contrast to the situation with electromagnetism, in which the spin \num{1} photon makes opposite charges attract each other, but similar charges repel each other. 

A follow-up question arises: why do we only observe ``gravitational charges'' (\ie, masses) of the same sign? The absence of negative masses is a prominent, but curious, feature of nature. One at first would expect some sort of symmetry between positive and negative masses.

The fact is that assuming that the equivalence principle holds in such cases (i.e., inertial and gravitational masses are equal), there is no such symmetry. The reason is the curious gravitational behavior of masses with different signs that can be understood in Newtonian gravity, for simplicity. Consider a pair of masses. If both the masses are positive, then the situation is standard and the masses will be attracted to each other. In particular, the gravitational force on each particle and their acceleration point in the same direction (see figure \ref{fig: signs-of-gravity}). If the masses have opposite signs, the forces are repulsive, but the net effect is that the negative mass pursues the positive mass, while the latter runs from the former. This is because the force and acceleration on the negative mass point to different directions (\(\vb{F} = - \abs{m} \vb{a}\)). Finally, for a similar reason, the gravitational force between two negative masses is attractive, but the effect is repulsive. This is pictured in figure \ref{fig: signs-of-gravity}. These examples make it clear that negative masses' gravitational dynamics are fairly different from positive masses. While negative and positive masses can lead to similar orbits in Newtonian mechanics \cite{shatskiy2011KeplerProblemCollisions}, we know of no experimental evidence supporting the existence of negative masses.

\begin{figure}[t]
    \centering
    \begin{tikzpicture}
        \coordinate (A) at (-1.5,0);
        \coordinate (B) at (+1.5,0);
        \coordinate (C) at (0,1ex);
        \draw[-{Stealth[round,length=3ex]},ultra thick] ($(B) + (C)$) -- ++(-1.25,0);
        \draw[-{Stealth[round,length=3ex]},ultra thick] ($(A) + (C)$) -- ++(+1.25,0);
        \draw[-{Stealth[round,length=3ex]},ultra thick] ($(B) - (C)$) -- ++(-1.25,0);
        \draw[-{Stealth[round,length=3ex]},ultra thick] ($(A) - (C)$) -- ++(+1.25,0);
        \node[anchor=south] at ($(B)+(C)+0.5*(-1.25,0)$) {\(\vb{F}\)};
        \node[anchor=south] at ($(A)+(C)+0.5*(+1.25,0)$) {\(\vb{F}\)};
        \node[anchor=north] at ($(B)-(C)+0.5*(-1.25,0)$) {\(\vb{a}\)};
        \node[anchor=north] at ($(A)-(C)+0.5*(+1.25,0)$) {\(\vb{a}\)};
        \node[white,font=\huge] at (A) {\faCircle};
        \node[white,font=\huge] at (B) {\faCircle};
        \node[NickBlue,font=\huge] at (A) {\faPlusCircle};
        \node[NickBlue,font=\huge] at (B) {\faPlusCircle};

        \coordinate (D) at ($(A)+(0,-1.5)$);
        \coordinate (E) at ($(B)+(0,-1.5)$);
        \coordinate (C) at (0,1ex);
        \draw[-{Stealth[round,length=3ex]},ultra thick] ($(E)$) -- ++(+1.25,0);
        \draw[-{Stealth[round,length=3ex]},ultra thick] ($(D) + (C)$) -- ++(-1.25,0);
        \draw[-{Stealth[round,length=3ex]},ultra thick] ($(E)$) -- ++(-1.25,0);
        \draw[-{Stealth[round,length=3ex]},ultra thick] ($(D) - (C)$) -- ++(-1.25,0);
        \node[anchor=south] at ($(E)+0.5*(+1.25,0)$) {\(\vb{F}\)};
        \node[anchor=south] at ($(D)+(C)+0.5*(-1.25,0)$) {\(\vb{F}\)};
        \node[anchor=south] at ($(E)+0.5*(-1.25,0)$) {\(\vb{a}\)};
        \node[anchor=north] at ($(D)-(C)+0.5*(-1.25,0)$) {\(\vb{a}\)};
        \node[white,font=\huge] at (D) {\faCircle};
        \node[white,font=\huge] at (E) {\faCircle};
        \node[NickBlue,font=\huge] at (D) {\faPlusCircle};
        \node[NickRed,font=\huge] at (E) {\faMinusCircle};
        
        \coordinate (F) at ($(D)+(0,-1.5)$);
        \coordinate (G) at ($(E)+(0,-1.5)$);
        \draw[-{Stealth[round,length=3ex]},ultra thick] ($(G)$) -- ++(+1.25,0);
        \draw[-{Stealth[round,length=3ex]},ultra thick] ($(F)$) -- ++(+1.25,0);
        \draw[-{Stealth[round,length=3ex]},ultra thick] ($(G)$) -- ++(-1.25,0);
        \draw[-{Stealth[round,length=3ex]},ultra thick] ($(F)$) -- ++(-1.25,0);
        \node[anchor=south] at ($(G)+0.5*(-1.25,0)$) {\(\vb{F}\)};
        \node[anchor=south] at ($(F)+0.5*(+1.25,0)$) {\(\vb{F}\)};
        \node[anchor=south] at ($(G)+0.5*(+1.25,0)$) {\(\vb{a}\)};
        \node[anchor=south] at ($(F)+0.5*(-1.25,0)$) {\(\vb{a}\)};
        \node[white,font=\huge] at (F) {\faCircle};
        \node[white,font=\huge] at (G) {\faCircle};
        \node[NickRed,font=\huge] at (F) {\faMinusCircle};
        \node[NickRed,font=\huge] at (G) {\faMinusCircle};
    \end{tikzpicture}
    \caption{Gravitational dynamics of two pointlike masses with various signs. Due to Newton's second law, a negative mass accelerates in the direction opposite to that of the applied force. Top: two positive masses accelerate toward each other. Middle: a positive mass accelerates away from a negative mass, but the negative mass accelerates toward the positive mass. Bottom: two negative masses accelerate away from each other.}
    \label{fig: signs-of-gravity}
\end{figure}
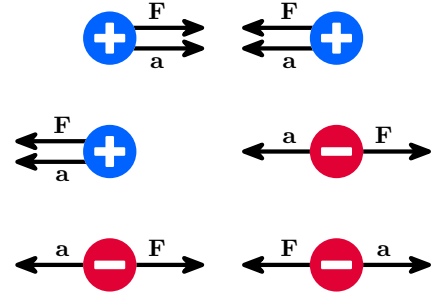

Within classical general relativity, this makes sense. One has a limited amount of matter types (\ie, of stress-energy-momentum tensors) that can be considered in the theory. These correspond to the fields of the standard model of particle physics and their emergent consequences. It turns out that all known classical forms of matter satisfy certain energy conditions, which are impositions made on the stress tensor to restrict the allowed behavior of the matter in a certain spacetime. Within classical general relativity, it is well-known that the so-called dominant energy condition (DEC), together with other reasonable assumptions, is sufficient to ensure the positivity of the total mass in an asymptotically flat spacetime at a given time \cite{schoen1979ProofPositiveMass,schoen1981ProofPositiveMass,witten1981NewProofPositive}. Hence, classical general relativity provides an answer to why we do not observe negative masses: as long as the matter content obeys the classical energy conditions, it follows that the total mass must be positive.

It happens, however, that quantum mechanical systems can easily violate the classical energy conditions, including the DEC. This means that these systems admit regions with negative energy density. For example, the Casimir effect \cite{casimir1948AttractionTwoPerfectly} leads to a negative energy density between two conducting plates. It is then natural to wonder whether one could produce a macroscopic object with negative mass by exploiting quantum effects. 

With this in mind, Costa and Matsas recently analyzed whether quantum mechanical effects could lead to macroscopic negative masses by considering the Casimir system \cite{costa2022CanQuantumMechanics}. They noticed that to keep the Casimir plates separated, it is necessary to hold them apart with ``struts.'' Assuming these struts are classical, they must satisfy the DEC, and the authors showed that the positive mass due to the struts is sufficient for the complete system to have a positive mass.

While Costa and Matsas analyzed the Casimir effect in Minkowski spacetime, their conclusions suggested the conjecture (already stated at the end of Ref.~\cite{costa2022CanQuantumMechanics}) that a ``cosmic-weight watcher must rule out from nature regular asymptotically flat stationary solutions of Einstein’s equations with \(M\) < 0''. This statement is vague regarding what could count as a ``cosmic-weight watcher,'' but some other conditions such as restrictions on the matter content or stability are known to be needed. For example, in Ref.~\cite{novikov2018StarsCreatingGravitational}, Novikov, Bisnovatyi-Kogan, and Novikov considered a number of possible stellar solutions with negative energy density. Their interest was mostly in the fact that a star with negative energy density tends to expand due to gravitational effects and contract due to the star's pressure---which is the exact opposite of the roles pressure and gravity play in regular stars. With this goal, the authors studied a few different models that show how general relativity allows these solutions, but do not comment on whether the matter necessary to form these stars exists or whether these solutions are stable. 

If one decides to take input from the matter theory, then there is evidence that negative masses should be forbidden. While quantum fields do not obey the classical energy conditions, there is evidence that they obey weakened energy conditions, which hold only on average. In rough terms, while quantum field theory allows negative energy densities in a given region, these negative energy densities must be balanced by positive energy densities elsewhere. With this assumption, and assuming that the Einstein field equations hold semiclassically, one can conclude under some geometric assumptions on the underlying spacetime that mass must always be positive. This conclusion is mostly supported by a theorem due to Penrose, Sorkin, and Woolgar \cite{penrose1993PositiveMassTheorem}, which must be complemented by a theorem due to Borde \cite{borde1987GeodesicFocusingEnergy} to translate the result into a condition about the stress tensor.

It is interesting, however, to return attention to the original cosmic-weight watcher conjecture and ask: can general relativity discard negative mass solutions without the aid of quantum field theory? This would provide an alternative mechanism for the nonoccurrence of negative masses that does not rely on quantum theory at all. Hence, if true, it would show that gravity itself abhors negative masses, regardless of the matter content.

A natural candidate for the role of a cosmic-weight watcher is the stability of solutions. More specifically, we conjecture that there are no stable regular asymptotically flat stationary solutions of Einstein’s equations with \(M\) < 0. This time, we consider the possibility that equilibrium solutions are allowed but discard their relevance based on whether or not they are dynamically stable. Our findings provide strong evidence for this refined conjecture. Namely, all models of negative-mass stars we considered turned out to be unstable under linear perturbations or presented other problems. \emph{A fortiori}, we establish that any stars with \(\qty(\pdv*{P}{\rho})_s < 0\) at any point must be dynamically unstable, and all barotropic negative-mass stars must satisfy this condition somewhere (and therefore are unstable). This analysis does not require the use of energy conditions, and thus the conclusion is very different from the earlier results on positive mass theorems.

The structure of the paper is as follows. Section \ref{sec: tov} reviews the Tolman--Oppenheimer--Volkoff equation, which is the basic equation for the hydrostatic equilibrium of a static and spherically symmetric star in general relativity. Section \ref{sec: energy-conditions} reviews, for completeness, the role of energy conditions in ruling out negative masses, with particular emphasis on the semiclassical scenario, in which the averaged null energy condition (ANEC) is relevant. Section \ref{sec: stability} discusses how to study the stability of static, spherically symmetric stars in general relativity and how to adapt that formalism to stars with negative masses. Section \ref{sec: examples} then provides concrete examples of negative-mass stars. Section \ref{sec: stability-negative-mass} gives our main results on how negative-mass stars are hydrodynamically unstable. We conclude in section \ref{sec: conclusions}. Appendix \ref{app: anec-stellar} expresses the ANEC integral (which is essential for verifying whether a given spacetime satisfies the ANEC) in a simpler form for a TOV-like star. Appendices \ref{app: nec-violations-qft} and \ref{app: bump-functions} are included solely in the \texttt{arXiv} version, and they discuss some minor details---the former exemplifies how quantum effects allow the violation of energy conditions in a simple way, while the latter deals with the details of integrals involving bump functions.

We employ the same conventions used in Ref. \cite{wald1984GeneralRelativity}, which includes abstract index notation and has the sign convention \(+++\) in the Misner--Thorne--Wheeler classification \cite{misner2017Gravitation}. Latin indices stand for abstract indices. We also use geometric units with \(G=c=1\). 

\section{Spherically Symmetric Equilibrium Configurations}\label{sec: tov}
To study negative-mass stars, we work on the framework of standard general relativity with some simplifying symmetry assumptions about the source \(\tensor{T}{_a_b}\) and metric. At this point, we refrain from interpreting its origin.

The first symmetry we demand is that the spacetime is stationary, which is motivated by the fact we would like to describe a negative-mass star that can retain equilibrium. Then we ask for spherical symmetry to keep the calculations simple and manageable. Finally, we impose that the material composing the star is isotropic, meaning all of the three principal pressures of the stress tensor coincide. This simplifies the analysis, but this condition could likely be lifted at the expense of the calculations becoming more complex. Finally, we shall also focus on finite configurations, hence neglecting cosmological scenarios.

Within these assumptions, there is a spherical coordinate system in which the line element can be written in the form
\begin{equation}\label{eq: static-spherically-symmetric-ansatz}
    \dd{s}^2 = - e^{2\phi(r)}\dd{t}^2 + e^{2\psi(r)}\dd{r}^2 + r^2 \dd{\Omega}^2,
\end{equation}
where \(\dd{\Omega}^2\) is the line element for the round metric in the unit two-sphere. \(r\) is such that the area of the spheres with constant coordinate radius \(r\) is \(4 \pi r^2\).

The symmetry assumptions also restrict the stress tensor to have the form
\begin{equation}\label{eq: fluid-stress-tensor}
    \tensor{T}{_a_b} = \rho(r) \tensor{u}{_a}\tensor{u}{_b} + P(r) (\tensor{g}{_a_b} + \tensor{u}{_a}\tensor{u}{_b}),
\end{equation}
where
\begin{equation}\label{eq: four-velocity-tov}
    \tensor{u}{^a} = e^{-\phi(r)} \tensor{\qty(\pdv{t})}{^a}.
\end{equation}
Notice that \(\tensor{u}{^a}\) is a normalized timelike vector which is everywhere parallel to the stationary Killing vector field. 

The problem of solving the Einstein field equations with these ansätze was originally considered in these coordinates by Tolman, Oppenheimer, and Volkoff \cite{tolman1934EffectInhomogeneityCosmological,tolman1934RelativityThermodynamicsCosmology,tolman1939StaticSolutionsEinstein,oppenheimer1939MassiveNeutronCores} and is reviewed in standard textbooks in general relativity \cite{wald1984GeneralRelativity,choquet-bruhat2015IntroductionGeneralRelativity,chrusciel2019ElementsGeneralRelativity,misner2017Gravitation} and stellar structure \cite{glendenning1997CompactStarsNuclear,zeldovich1996StarsRelativity}. It can be shown that the solution inside the star is given by the line element
\begin{equation}
    \dd{s}^2_< = - e^{2\phi(r)}\dd{t}^2 + \qty(1 - \frac{2m(r)}{r})^{-1}\dd{r}^2 + r^2 \dd{\Omega}^2,
    \label{eq: staticsphericallineel}
\end{equation}
where the subscript \(<\) in \(\dd{s}^2_<\) indicates this is the interior solution. The functions \(\phi\) and \(m\) are determined by means of the differential equations
\begin{align}
    \dv{m}{r} &= 4 \pi \rho(r) r^2, \label{eq: dm-dr} \\
    \dv{\phi}{r} &= \frac{4 \pi P(r) r^3 + m(r)}{r[r - 2m(r)]}, \label{eq: dphi-dr} \\
    \dv{P}{r} &= -(\rho + P) \frac{4 \pi P(r) r^3 + m(r)}{r[r - 2m(r)]}. \label{eq: tov}
\end{align}
Equation (\ref{eq: tov}) is known as the Tolman--Oppenheimer--Volkoff (TOV) equation.

The system of differential equations given by Eqs. (\ref{eq: dm-dr}) to (\ref{eq: tov}) is underdetermined, since it involves four unknowns (\(\rho\), \(P\), \(m\), and \(\phi\)), but only three equations. In most applications, the remaining equation is supplied in the form of a thermodynamical equation of state relating \(P\) and \(\rho\). The procedure to solve the TOV system is then to integrate Eqs. (\ref{eq: dm-dr}) and (\ref{eq: tov}) from the center to the border of the star with the aid of the equation of state. This is often done numerically by providing the initial conditions \(m(0) = 0\) and \(P(0) = P_0\) \footnote{It is also possible to provide \(\rho(0) = \rho_0\) instead of \(P(0) = P_0\). The condition \(m(0) = 0\) is used to avoid the presence of a physical singularity at the origin (this singularity would be analogous to the singularity at the origin of Schwarzschild spacetime). Notice also that in practical numerical computations, one does not give the initial condition precisely at \(r = 0\), where the right-hand side of the TOV equation is indeterminate at best. Rather, one introduces a small cutoff \(r_0 > 0\) and provides initial conditions at \(r_0\). Here, we use this method with initial conditions of the form \(\rho(r_0) = \rho_0\) and \(m(r_0) = \frac{4}{3}\pi \rho_0 r_0^3\).}. One then proceeds with the integration until the boundary of the star is reached at \(r=R\), defined by the condition that 
\begin{equation}
    P(R) = 0. \label{eq: P-at-R}
\end{equation}
Then, it is usually assumed that the star ends at \(R\) and \(\rho\) and \(P\) are understood to vanish for \(r > R\). In this outer region, the metric is given by the standard Schwarzschild metric with mass parameter
\begin{equation}
    M = m(R). \label{eq: m-at-R}
\end{equation}

Once \(\rho\), \(P\), and \(m\) are known, one can solve Eq. (\ref{eq: dphi-dr}) to obtain \(\phi\). The boundary condition is then that the metric of the spacetime is continuous across the stellar boundary at \(r = R\) and translates into the condition \footnote{In numerical computations, it may be more efficient to solve for \(\phi\) while solving for the remaining variables. This can be done by imposing any initial condition for \(\phi\) at the center of the star, and then using the fact that Eq. (\ref{eq: dphi-dr}) is linear in \(\phi\) to subtract the value obtained at the boundary and add the value desired at the boundary, hence fixing the appropriate boundary condition.}
\begin{equation}
    e^{2\phi(R)} = \qty(1 - \frac{2M}{R}). \label{eq: phi-at-R}
\end{equation}

The boundary conditions at \(R\), Eqs. (\ref{eq: P-at-R}) to (\ref{eq: phi-at-R}), can be deduced from the Israel junction conditions \cite{israel1966SingularHypersurfacesThin,*israel1967SingularHypersurfacesThin} with the additional assumption that the stress-energy tensor is non-singular at the stellar surface. Stars with discontinuous pressure can be described by allowing a thin matter shell at \(r=R\), in which case the boundary conditions are relaxed as well.

In most applications in astrophysics, one is interested in a particular model for a star and wants to understand how this model gravitates. For example, one can consider an equation of state modeling a neutron star and use it to solve the TOV system. Nevertheless, this approach is restrictive in the study of negative-mass stars, whose equation of state is not predetermined. Instead of using only equations of state to fix a stellar solution, one can impose that the star should have a particular energy density profile. In other words, one imposes by hand that the function \(\rho(r)\) is given by an ansatz. With this function fixed, one then solves the TOV system with the condition that the star ends at some predefined radius \(r=R\). This method was employed in Ref. \cite{novikov2018StarsCreatingGravitational}. 

\subsection{Example: The Schwarzschild Star}
    A first example of a negative-mass star is based on a star of constant density profile \(\rho(r) = \rho_0\). This is known as a Schwarzschild star \cite{schwarzschild1916UberGravitationsfeldKugel}. Under this ansatz, the TOV system can be handled analytically. The mass parameter is given by
    \begin{equation}
        m(r) = \frac{4 \pi \rho_0 r^3}{3},
    \end{equation}
    and the boundary condition \(m(R) = M\) fixes the value of \(M\) according to
    \begin{equation}
        \rho_0 = \frac{3 M}{4 \pi R^3}.
    \end{equation}
    We are interested in the case with \(M < 0\) (and hence \(\rho_0 < 0\)).

    The solution for the pressure is the same as in the positive-mass case. It is given by 
    \begin{equation}\label{eq: schwarzschild-pressure}
        P(r) = \rho_0 \qty[\frac{\sqrt{1 - \frac{2M}{R}} - \sqrt{1 - \frac{2M r^2}{R^3}}}{\sqrt{1 - \frac{2M r^2}{R^3}} - 3 \sqrt{1 - \frac{2M}{R}}}].
    \end{equation}
    For positive values of \(M\), the pressure is everywhere non-negative and finite if \(M/R < 4/9\), which holds for a wide class of equations of state and is known as the Buchdahl limit \cite{buchdahl1959GeneralRelativisticFluid}. Notice that the pressure is still everywhere non-negative and finite even if the star has negative mass. 

    Having found a star with negative mass, we address two important questions in the next two sections:
    \begin{itemize}
        \item[i.] Is the matter composing this star allowed by quantum field theory?
        \item[ii.] Is this configuration stable?
    \end{itemize}

\section{Energy Conditions}\label{sec: energy-conditions}
It is not surprising that the Einstein equations allow for negative-mass stars. After all, any Lorentzian geometry is a solution to the Einstein equations as long as one chooses the right stress-energy-momentum tensor. 

To put constraints on the physical reasonableness of a given solution, one often imposes energy conditions \cite{curiel2017PrimerEnergyConditions,fewster2017QuantumEnergyInequalities,kontou2020EnergyConditionsGeneral,martin-moruno2017ClassicalSemiclassicalEnergy,witten2020LightRaysSingularities}. These are restrictions on the stress tensor that enforce it to have some interesting properties, typically understood as energy being positive in a suitable sense. The weakest of the classical energy conditions is the null energy condition (NEC), reviewed for example in Refs. \cite{curiel2017PrimerEnergyConditions,kontou2020EnergyConditionsGeneral,martin-moruno2017ClassicalSemiclassicalEnergy,witten2020LightRaysSingularities}. It states that for all null vectors \(\tensor{k}{^a}\) the stress tensor satisfies the bound
\begin{equation}
    \tensor{T}{_a_b}\tensor{k}{^a}\tensor{k}{^b} \geq 0.
\end{equation}
In the particular case of interest, a perfect fluid, this condition states that \(\rho + P \geq 0\). 

The NEC is not the only energy condition of classical interest. It is important to point out the dominant energy condition (DEC), which states that for all future-directed causal vectors \(\tensor{\xi}{^a}\) it holds that
\begin{equation}
    -\tensor{T}{^a_b}\tensor{\xi}{^b}
\end{equation}
is causal and future-directed. This condition essentially requires all observers to see causal, future-directed energy fluxes. For a perfect fluid, the DEC states that \(\rho \geq \abs{P}\). Notice that the DEC implies the NEC. The interest in this condition is that it was used by Schoen, Yau, and Witten to obtain the first positive-mass theorems in general relativity \cite{schoen1979ProofPositiveMass,schoen1981ProofPositiveMass,witten1981NewProofPositive}. Hence, within suitable assumptions, the dominant energy condition ensures that the spacetime has a positive mass.

All known forms of classical matter satisfy the NEC \footnote{It is known that some forms of classical matter, such as a non-minimally coupled scalar field, violate the NEC---see, \eg, Refs. \cite{fewster2017QuantumEnergyInequalities,kontou2020EnergyConditionsGeneral}. Nevertheless, we are restricting our focus to the classical fields associated with the standard model, and we assume general relativity to be the theory describing gravitational interactions. In particular, once we assume gravity to be described by general relativity, we also assume that all standard model fields are minimally coupled. In any case, this does not affect the main points of our discussion.}, and even the DEC. Does our star?

The positive mass theorems due to Schoen, Yau, and Witten already imply that the star cannot satisfy the DEC. Using Eq. (\ref{eq: schwarzschild-pressure}), one can promptly show that
\begin{equation}
    \rho(r) + P(r) = \rho_0 \qty[\frac{2 \sqrt{1 - \frac{2M}{R}}}{3 \sqrt{1 - \frac{2M}{R}} - \sqrt{1 - \frac{2Mr^2}{R^3}}}].
\end{equation}
The term in brackets is always positive for \(0 \leq r \leq R\). Since \(\rho_0 < 0\), we conclude that \(\rho(r) + P(r) < 0\). Hence, the star cannot be built from any known form of classical matter.

This is not surprising. After all, we started this discussion by considering the Casimir effect, which is inherently quantum. There is no \emph{a priori} reason to expect that classical configurations of matter could lead to negative energy solutions, but, in principle, one may question whether quantum configurations can lead to macroscopic negative masses since they are known to allow local violations of the positivity of energy---in fact, all classical energy conditions, including the NEC, can be easily violated by exploiting quantum effects~\cite{fewster2012LecturesQuantumEnergy,fewster2017QuantumEnergyInequalities}. See App. \ref{app: nec-violations-qft}. 

A quantum analog of the classical energy conditions is the averaged null energy condition (ANEC). In the conventions of Ref.~\cite{kontou2020EnergyConditionsGeneral}, the ANEC states that, for any inextendible null geodesic \(\gamma\) with affine parameter \(\lambda\),
\begin{equation}\label{eq: anec-integral}
    \int_{-\infty}^{+\infty} \tensor{T}{_a_b}(\gamma(\lambda)) \tensor{\dot{\gamma}}{^a}(\lambda) \tensor{\dot{\gamma}}{^b}(\lambda) \dd{\lambda} \geq 0
\end{equation}
whenever the integral is absolutely convergent. \(\tensor{T}{_a_b}\) is now understood as the renormalized expectation value of the stress tensor. The ANEC states that the NEC can be locally violated, but it still holds on average---an idea originally due to Tipler \cite{tipler1978EnergyConditionsSpacetime}. An even weaker condition is the achronal ANEC (AANEC), which states that the ANEC must hold only for achronal geodesics, \ie, for null geodesics such that no pair of its points can be connected by a timelike curve. The restriction of the ANEC to achronal geodesics was originally considered by Wald and Yurtsever \cite{wald1991GeneralProofAveraged}, and it is interesting because there is evidence that the AANEC may be true in all physically reasonable circumstances and that it could be implied by a fundamental principle of full quantum gravity---see Refs. \cite{kontou2020EnergyConditionsGeneral,wall2010ProvingAchronalAveraged}. Furthermore, the ANEC is equivalent to the quantum null energy condition \cite{bousso2016ProofQuantumNull,ceyhan2020RecoveringQNECANEC}, which is interesting in its own right and is implied in turn by an interesting conjecture pertaining to the quantum focusing of congruences of geodesics \cite{bousso2016QuantumFocusingConjecture}.

The fact that the ANEC demands that energy be positive on average is remarkably similar to how Costa and Matsas dealt with the Casimir system: while the energy density could be negative within the plates, it was positive somewhere else, and the net energy was positive. Hence, the ANEC appears to provide an interesting criterion for establishing the net positivity of mass. Instead of focusing on forcing the energy density to be positive everywhere, we simply require that negative energy densities be compensated elsewhere. Furthermore, it should be mentioned that the ANEC does hold in the Casimir system \cite{fewster2007AveragedNullEnergy,graham2005PlateHoleObeys}.

Assuming that every AANEC integral is absolutely convergent, this intuition is correct. This is a corollary of works by Borde \cite{borde1987GeodesicFocusingEnergy} and Penrose, Sorkin, and Woolgar \cite{penrose1993PositiveMassTheorem}, which we briefly review below. 

Borde's theorem is a statement about the occurrence of conjugate points in geodesics satisfying a weaker form of an averaged energy condition. His original result, focusing theorem 2 in Ref. \cite{borde1987GeodesicFocusingEnergy}, applies to any causal geodesic, but we will focus on the case of null geodesics. Borde's theorem is the following. 

\begin{theorem}[Borde \cite{borde1987GeodesicFocusingEnergy}]
    Consider a complete null geodesic with tangent vector \(\tensor{k}{^a}\) and affine parameter \(\lambda\). For each \(\epsilon > 0\), assume there is some \(b > 0\) such that for any \(\lambda_1 < \lambda_2\) there is a pair of intervals \(I^- < \lambda_1\) and \(I^+ > \lambda_2\)---each with length larger than or equal to \(b\)---such that
    \begin{equation}
        \int_{\lambda'}^{\lambda''} \tensor{R}{_a_b}\tensor{k}{^a}\tensor{k}{^b} \dd{\lambda} \geq - \epsilon\qc \forall{} \lambda' \in I^-, \forall{} \lambda'' \in I^+.
    \end{equation}
    If \(\tensor{k}{^c}\tensor{k}{^d}\tensor{k}{_[_a}\tensor{R}{_b_]_c_d_[_e}\tensor{k}{_f_]} \neq 0\) at some point on \(\gamma\), then \(\gamma\) contains a pair of conjugate points.
\end{theorem}

This is a generalization of an earlier result due to Hawking and Penrose that establishes the occurrence of conjugate points based on the validity of the null convergence condition, which states that \(\tensor{R}{_a_b}\tensor{k}{^a}\tensor{k}{^b} \geq 0\) for all null vectors \(\tensor{k}{^a}\) \cite{hawking1970SingularitiesGravitationalCollapse}. The basic idea of both proofs is to use the curvature condition to ensure that the geodesics do not defocus, while the condition \(\tensor{k}{^c}\tensor{k}{^d}\tensor{k}{_[_a}\tensor{R}{_b_]_c_d_[_e}\tensor{k}{_f_]} \neq 0\) ensures that the geodesic interacts with curvature at some point for the focusing to start. The main technique in the proof consists of manipulating the Raychaudhuri equation \cite{raychaudhuri1955RelativisticCosmology}. 

Under the hypothesis of absolute convergence of the integral $\int_{-\infty}^\infty R_{ab}k^ak^b\dd\lambda$, we can show the following intuitive corollary.

\begin{corollary}
    Consider a complete null geodesic with tangent vector \(\tensor{k}{^a}\) and affine parameter \(\lambda\). Suppose that
    \begin{equation}
        \int_{-\infty}^{+\infty} \tensor{R}{_a_b}\tensor{k}{^a}\tensor{k}{^b} \dd{\lambda} \geq 0
    \end{equation}
    and that the integral converges absolutely. If it holds that \(\tensor{k}{^c}\tensor{k}{^d}\tensor{k}{_[_a}\tensor{R}{_b_]_c_d_[_e}\tensor{k}{_f_]} \neq 0\) at some point on \(\gamma\), then \(\gamma\) must contain a pair of conjugate points.
\end{corollary}

The assumption of absolute convergence ensures that the ``tails'' of the integrand as \(\abs{\lambda} \to +\infty\) are negligible compared to the ``bulk'' contributions. Notice that if the Einstein field equations hold (as we assume they do), then the condition on the integral of \(\tensor{R}{_a_b}\tensor{k}{^a}\tensor{k}{^b}\) is equivalent to the ANEC integral being non-negative. Furthermore, it is known that achronal null geodesics cannot have conjugate points (Proposition 4.5.12 in Ref. \cite{hawking1973LargeScaleStructure}). Hence, we get the following result.

\begin{corollary}\label{corol: final-result-borde}
    Consider a complete null geodesic with tangent vector \(\tensor{k}{^a}\) and affine parameter \(\lambda\). Suppose that the Einstein equations and the ANEC hold. If \(\tensor{k}{^c}\tensor{k}{^d}\tensor{k}{_[_a}\tensor{R}{_b_]_c_d_[_e}\tensor{k}{_f_]} \neq 0\) at some point on \(\gamma\), then \(\gamma\) is not achronal.
\end{corollary}

This was known to Penrose, Sorkin, and Woolgar at the time of writing of Ref. \cite{penrose1993PositiveMassTheorem}. The result they established is then a theorem concerning the positivity of mass based on the existence of achronal geodesics in the spacetime. 

The basic idea behind the argument runs as follows. The Shapiro time delay \cite{shapiro1964FourthTestGeneral} shows that null geodesics that pass closer to a positive mass are delayed (in coordinate time) relative to geodesics that pass farther away. Analogously, one concludes that geodesics passing close to a negative mass will be ``faster'' (in coordinate time) than geodesics that pass far from the negative mass. With this in mind, Penrose, Sorkin, and Woolgar constructed the fastest causal curve from past null infinity to future null infinity of the conformal completion of an asymptotically flat spacetime and showed this is an achronal geodesic. This geodesic can pass either through the interior of the spacetime or avoid it completely by traveling only through the conformal boundaries of the spacetime. Due to the Shapiro time delay, if the spacetime has negative mass, this geodesic must enter the spacetime. They then conclude that if the spacetime has no complete achronal null geodesics, the spacetime mass must be non-negative. 

The precise statement is as follows \cite{chrusciel2004PoorManPositive,cameron2023PositivityMassHigher}:

\begin{definition}[Uniformly Schwarzschildean]
    Let \(M \in \mathbb{R}\) and let \(B \subseteq \mathbb{R}^{3}\) be a ball with radius \(R > 2 M\). Consider the metric \(\tensor*{g}{^M_a_b}\) on \(\mathbb{R} \times (\mathbb{R}^3 \setminus B)\) with line element
    \begin{equation}
        \dd{s}_M^2 = - \qty(1 - \frac{2M}{r}) \dd{t}^2 + \qty(1 - \frac{2M}{r})^{-1} \dd{r}^2 + r^2 \dd{\Omega}^2. \label{eq: sch-M-sph}
    \end{equation}
    A metric \(\tensor{g}{_a_b}\) on \(\mathbb{R} \times (\mathbb{R}^3 \setminus B)\) is said to be uniformly Schwarzschildean if there is an \(M \in \mathbb{R}\) such that, in the coordinates of Eq. (\ref{eq: sch-M-sph}), it holds that
    \begin{gather}
        \tensor{g}{_\mu_\nu} - \tensor*{g}{^M_\mu_\nu} = o\qty(\abs{M} r^{-1}), \\
        \tensor{\partial}{_\mu}(\tensor{g}{_\nu_\rho} - \tensor*{g}{^M_\nu_\rho}) = o\qty(\abs{M} r^{-2}).
    \end{gather}
    In the case \(M = 0\), the above notation is understood to mean the metric is flat for \(r > R\), where \(R \geq 0\) is some constant.
\end{definition}

\begin{theorem}[Penrose--Sorkin--Woolgar \cite{penrose1993PositiveMassTheorem}]
    Consider a spacetime \((M,g)\) taken to be asymptotically flat at null and spatial infinity and uniformly Schwarzschildean. Let \(\mathcal{D} = I^-(\mathcal{I}^+) \cap I^+(\mathcal{I}^-)\) be the domain of outer communications of the spacetime. Assume that \(\mathcal{D} \cup \mathcal{I}^+ \cup \mathcal{I}^-\) is globally hyperbolic as a subset of the conformal extension of \(M\). If there are no achronal null geodesics connecting \(\mathcal{I}^-\) to \(\mathcal{I}^+\) and passing through \(\mathcal{D}\), then it follows that the ADM four-momentum of \(M\) is future-causal.  
\end{theorem}

Recall that the ADM four-momentum (after Arnowitt, Deser, and Misner \cite{arnowitt2008RepublicationDynamicsGeneral}) is a notion of four-momentum of the whole spacetime measured at a single instant of time (which is to be understood as a spacelike hypersurface of the spacetime). 

The Penrose--Sorkin--Woolgar theorem can then be combined with corollary \ref{corol: final-result-borde} to yield the following result. 

\begin{corollary}
    Consider a spacetime \((M,g)\) taken to be asymptotically flat at null and spatial infinity and uniformly Schwarzschildean. Let \(\mathcal{D} = I^-(\mathcal{I}^+) \cap I^+(\mathcal{I}^-)\) be the domain of outer communications of the spacetime. Assume that \(\mathcal{D} \cup \mathcal{I}^+ \cup \mathcal{I}^-\) is globally hyperbolic as a subset of the conformal extension of \(M\). Suppose the Einstein field equations hold, that the AANEC integrals are always absolutely convergent and non-negative, and that the null generic condition holds---\ie, that for all null geodesics, there is some point at which \(\tensor{k}{^c}\tensor{k}{^d}\tensor{k}{_[_a}\tensor{R}{_b_]_c_d_[_e}\tensor{k}{_f_]} \neq 0\). Then the ADM four-momentum of the spacetime is future-causal.
\end{corollary}

In the above result, the AANEC is sufficient (as opposed to the ANEC), because we only need to rule out the achronal geodesics. The AANEC will typically hold as a vacuous truth.

The exterior solution for a spherically symmetric star in general relativity is always the Schwarzschild solution due to Birkhoff's theorem. Hence, all spherically symmetric stellar spacetimes must be uniformly Schwarzschildean, and hence satisfy the conditions for the Penrose--Sorkin--Woolgar theorem. It follows that a star satisfying the AANEC and the null generic condition can never have negative mass in general relativity. Due to the strong evidence supporting the AANEC \cite{wall2010ProvingAchronalAveraged,kontou2020EnergyConditionsGeneral}, it seems unlikely that one can build negative-mass stars within the domain of the semiclassical Einstein equations. 

\section{Dynamical Stability of Stars}\label{sec: stability}
From the perspective of averaged energy conditions, it seems unlikely that one can build negative-mass stars. However, this requires a lot of input from quantum theory, and it would be interesting to try to discard negative masses based purely on classical physics. With this in mind, this section studies the problem from the perspective of stability analysis. 

There are three obvious ways in which a star can be spontaneously destroyed: through hydrodynamic instabilities, thermodynamic instabilities, and the consumption of nuclear fuel. In regular stars, the time scales associated with hydrodynamic processes are much shorter than those associated with thermal processes, which, in turn, are much shorter than those associated with nuclear processes. Ref. \cite{zeldovich1996StarsRelativity}, for example, considers the Sun and estimates the time scales for hydrodynamic processes at around \(t_H \sim \SI{e3}{\second}\), for thermal processes at \(t_T \sim \SI{3e7}{\year}\), and \(t_N \sim \SI{e10}{\year}\) for the burning of nuclear fuel. Hence, hydrodynamic stability is the most important one in regular stars.

While the estimates for the Sun cannot be trusted when dealing with negative-mass stars, it still seems natural to focus on hydrodynamic stability. While this is partially motivated by the time scale estimates we just mentioned, there is also the important reason that a thermodynamic or nuclear analysis would require a precise model of the matter constituting the star, for example in the form of an equation of state. This model, however, is not available, since we are precisely avoiding making too many assumptions on the composition of matter. Hence, we have insufficient information to discuss the thermal or nuclear stabilities of negative-mass stars. Furthermore, since they would necessarily be made of unconventional matter (or purely of quantum effects), one lacks a good justification for the use of traditional techniques developed for regular matter. Thus, we focus on the hydrodynamic stability.

The simplest way of discussing the hydrodynamic stability of a star is to consider linear perturbations of the stellar parameters and metric coefficients about the background spacetime provided by the TOV solution. One then writes the metric and the stellar parameters as the equilibrium solutions to the TOV equations plus small time-dependent perturbations. The Einstein equations for these perturbations are linearized. The end goal is to see whether the perturbations grow in time or remain bounded. If they grow in time, the star is deemed unstable. 

This problem was originally considered by Chandrasekhar \cite{chandrasekhar1964DynamicalInstabilityPRL,*chandrasekhar1964DynamicalInstabilityPRLErratum,chandrasekhar1964DynamicalInstabilityApJ,*chandrasekhar1964DynamicalInstabilityApJErratum} (see also the review in Chap. 26 of Ref. \cite{misner2017Gravitation}), who simplified the problem to the study of solutions of the so-called Chandrasekhar pulsation equation. This is an equation that describes the Lagrangian displacement of the perturbation \(\xi\), which is a measure of how much each fluid element in the star is dislocated by the perturbation. If the fluid element was originally at \(r\) before the perturbation, it is taken to \(r+ \xi(t,r)\) by the perturbation. The pulsation equation is given by
\begin{equation}\label{eq: SLP-time}
    \pdv{r}\qty[p(r) \pdv{\chi}{r}] + q(r) \chi(t,r) = w(r) \pdv[2]{\chi}{t},
\end{equation}
where \(\chi\) is a ``renormalized Lagrangian displacement'', from which all remaining perturbations can be promptly calculated. It is given in terms of the ``true Lagrangian displacement'' \(\xi\) by 
\begin{equation}
    \chi(t,r) = r^2 e^{-\phi(r)} \xi(t,r).
\end{equation}
The coefficient functions \(p\) (not to be mistaken for the pressure \(P\)), \(q\), and \(r\) are determined by the background stellar spacetime according to
\begin{widetext}
\begin{align}
    p(r) &= e^{\psi + 3 \phi}\frac{\gamma P}{r^2}, \label{eq: SLP-p-original} \\
    q(r) &= e^{\psi + 3 \phi} \qty[\frac{1}{r^2(P + \rho)}\qty(\dv{P}{r})^2 - \frac{4}{r^3}\dv{P}{r} - 8\pi \frac{P}{r^2}  e^{2 \psi} (P+\rho)], \label{eq: SLP-q-original} \\
    w(r) &= \frac{(\rho + P)}{r^2} e^{3\psi + \phi}. \label{eq: SLP-w-original}
\end{align}
\end{widetext}
The coefficient \(p\) in the Chandrasekhar pulsation equation involves the ``effective polytropic index''
\begin{equation}\label{eq: gamma-definition}
    \gamma = \frac{1}{P} \qty(\pdv{P}{n})_\rho \qty[n - (\rho + P) \qty(\pdv{n}{\rho})_P],
\end{equation}
where \(n\) is the baryon number density (with antibaryons counted negatively).

Typically, the boundary conditions for the problem are that \cite{bardeen1966CatalogueMethodsStudying} 
\begin{equation}\label{eq: boundary-SLP-time}
    \lim_{r \to 0} \frac{\abs{\chi(t,r)}}{r^3} < \infty \qq{and} \qty[\frac{\gamma P e^{\phi}}{r^2} \pdv{\chi}{r}]_R = 0.
\end{equation}
The first one ensures \(\chi\) and its derivative stay finite at the center of the star, while the second one enforces that fluid elements on the boundary of the star stay at the boundary of the star \footnote{This is done by imposing that the Lagrangian perturbation of the pressure vanishes at the boundary.}. Additional conditions may be imposed to keep \(\xi\) finite at the boundary \cite{bardeen1966CatalogueMethodsStudying}.

The solution of Eq. (\ref{eq: SLP-time}) is carried out by the separation of variables \(\chi(t,r) = \tau(t)\zeta(r)\). The resulting equations are 
\begin{equation}\label{eq: SLP}
    \dv{r}\qty[p(r) \dv{\zeta}{r}] + q(r) \zeta(r) + \sigma^2 w(r) \zeta(r) = 0,
\end{equation}
also known as the Chandrasekhar pulsation equation, and
\begin{equation}\label{eq: SHO-Chandra}
    \dv[2]{\tau}{t} + \sigma^2 \tau = 0.
\end{equation}

The boundary conditions take the form
\begin{equation}\label{eq: boundary-SLP}
    \lim_{r \to 0} \frac{\abs{\zeta(r)}}{r^3} < \infty \qq{and} \qty[\frac{\gamma P e^{\phi}}{r^2} \dv{\zeta}{r}]_R = 0.
\end{equation}

The solution to Eq. (\ref{eq: SHO-Chandra}) is 
\begin{equation}
    \tau(t) = \tau(0) \cos(\sigma t) + \frac{\dot{\tau}(0)}{\sigma} \sin(\sigma t),
\end{equation}
where the dot denotes a time derivative. Notice, in particular, that if \(\sigma^2 < 0\) (\ie, if \(\sigma\) is imaginary), then \(\tau(t)\) grows exponentially in time for a generic initial condition. In this case, the perturbation is unstable, and so is the star. Hence, if one desires to study the stability of a star with the Chandrasekhar pulsation equation, the main question to be addressed is whether \(\sigma^2 > 0\) holds for all pulsation modes or whether there are exceptions. If there is a single admissible value of \(\sigma\) for which \(\sigma^2 < 0\), then the star is unstable because \(\tau(t)\) will grow exponentially in time for this mode.

Eq. (\ref{eq: SLP}), when combined with the boundary conditions (\ref{eq: boundary-SLP}), corresponds to a Sturm--Liouville problem \cite{arfken2013MathematicalMethodsPhysicists,teschl2012OrdinaryDifferentialEquations,zettl2005SturmLiouvilleTheory,pryce1993NumericalSolutionSturm} with eigenvalues \(\sigma^2\). This determines the allowed values of \(\sigma\). In standard astrophysical applications, it holds that \(p(r) > 0\) and \(w(r) > 0\) in \((0,R)\). This ensures the problem is a self-adjoint eigenvalue problem in a Hilbert space, and thus \(\sigma^2\) is always real.

The standard methods for studying stability with the Chandrasekhar pulsation equation involve techniques from Sturm--Liouville theory. For example, method 2-C in Ref. \cite{bardeen1966CatalogueMethodsStudying} obtains the number of unstable modes by numerically solving the Chandrasekhar equation with \(\sigma^2 = 0\) and counting the number of zeros in the solution. By the Sturm Comparison Theorem \cite{teschl2012OrdinaryDifferentialEquations,zettl2005SturmLiouvilleTheory}, this number of zeros can be understood as the number of eigenvalues smaller than zero, \ie, the number of eigenvalues with \(\sigma^2 < 0\). This is only possible because, for usual stars, \(p(r) > 0\) and \(w(r) > 0\) for \(r \in (0,R)\), which means the problem is sufficiently regular to use standard results from Sturm--Liouville theory (although the endpoints \(r=0\) and \(r=R\) are often singular). Alternatively, method 2-D in Ref. \cite{bardeen1966CatalogueMethodsStudying} uses a variational technique to obtain the sign of the smallest eigenvalue, which is a technique based on the so-called min-max principle of functional analysis \cite{pryce1993NumericalSolutionSturm,reed1978AnalysisOperators}.

There are a few reasons, however, that keep the standard Chandrasekhar pulsation equation from being straightforward when dealing with negative-mass stars. The first is that the NEC violations which are required in negative-mass stars (otherwise the Penrose--Sorkin--Woolgar theorem would imply the star has positive mass) lead to sign flips in the coefficients \(p\) and \(w\). This means we are apparently outside the domain of the regular Sturm--Liouville problem, and there is no \emph{a priori} reason to believe that results such as the Sturm Comparison theorem still hold. Hence, at first glance, it seems challenging to extract the sign of \(\sigma^2\) in a NEC-violating star.

The second reason is also practical in nature and involves the difficulty in knowing the equation of state for the star. Recall that the Chandrasekhar pulsation equation involves the effective polytropic index \(\gamma\). For an adiabatic perturbation, one can rewrite this expression as~\cite{merafina1989SystemsSelfgravitatingClassical}
\begin{equation}\label{eq: gamma-adiabatic}
    \gamma = \frac{\rho + P}{P} \qty(\pdv{P}{\rho})_s,
\end{equation}
where \(s\) is the entropy per baryon in the star. In the case of a barotropic equation of state (\ie, an equation of state of the form \(P = P(\rho)\)), we can write
\begin{equation}
    \qty(\pdv{P}{\rho})_s = \dv{P}{r}\qty(\dv{\rho}{r})^{-1}
\end{equation}
and compute \(\gamma\) as a function of \(r\) without using further information about the equation of state. However, for an equation of state that is not barotropic, this ratio of radial derivatives could correspond to the partial derivative with different quantities being held constant because \(\qty(\pdv*{P}{\rho})_T\), \(\qty(\pdv*{P}{\rho})_s\), and \(\qty(\pdv*{P}{\rho})_n\), for example, need not coincide. 

If the star in question was obtained from an equation of state, this is not a problem. However, stars obtained from energy density profiles provide us more freedom to exploit when building examples. For these stars, there is no \emph{a priori} reason to believe that its constituent matter is barotropic. Hence, the need to know \(\gamma\) is particularly inconvenient.

Finally, a third objection to the Chandrasekhar pulsation equation is conceptual. One of the assumptions that enter Chandrasekhar's derivation is baryon number conservation. Some authors go as far as saying that ``The most fundamental law of thermodynamics---even more fundamental than the `first' and `second' laws---is baryon conservation'' \cite[p. 558]{misner2017Gravitation}. Nevertheless, it is not clear whether this is applicable in the case of a negative-mass star, in which case, it is natural to begin by generalizing ``baryon number'' to ``particle number'', where particle refers to the particles of whichever field composes the star. Nevertheless, even with this generalization the conservation of particle number is still unclear. For example, a quantum state violating the NEC need not be an eigenstate of the particle number operator. Hence, the number of particles in a negative mass star is, in general, ill-defined. We thus prefer to avoid techniques which rely on baryon number conservation.

Luckily, within suitable assumptions, these difficulties can be overcome. The difficulty with \(\gamma\) can be overcome in the particular case of a barotropic equation of state, and the validity of this hypothesis can be verified by making a parametric plot of \(P(r)\) against \(\rho(r)\) and checking whether the resulting graph yields a function. While this is not proof that the equation of state is barotropic, it gives evidence that the star can be well-approximated by a barotropic equation of state. In practice, however, this equation of state may be unobtainable in closed form.

The difficulty concerning baryon number conservation can be solved in the particular case of an adiabatic perturbation. The hypothesis of baryon number conservation enters Chandrasekhar's derivation in order to express the Eulerian pressure perturbation \(\var{P}\) as a function of the Lagrangian displacement \(\xi\). For a general perturbation, this method seems adequate. For an adiabatic perturbation, however, we can impose simply that the Lagrangian pressure perturbation \(\Delta P\) satisfies
\begin{equation}
    \Delta P = \qty(\pdv{P}{\rho})_s \Delta \rho.
\end{equation}
With the aid of Eq. (\ref{eq: gamma-adiabatic}), one can show that this prescription yields the same result obtained by Chandrasekhar \cite{chandrasekhar1964DynamicalInstabilityApJ} using baryon number conservation. Recall that the Lagrangian perturbation \(\Delta \alpha\) of a quantity \(\alpha\) is the perturbative change in a quantity as one follows a fluid element, whereas the Eulerian perturbation \(\var\alpha\) is the change as one considers a fixed point in space \cite{misner2017Gravitation,thorne2017ModernClassicalPhysics}. They are related, to linear order, by
\begin{equation}
    \Delta \alpha(t,r) = \var\alpha(t,r) + \xi(t,r) \dv{\alpha}{r}\qty(r),
\end{equation}
where \(\dv*{\alpha}{r}\) refers to the background (unperturbed) value of \(\alpha\).

We notice that an incorrect approach one could consider to avoid these difficulties is to impose that the perturbations satisfy a simple equation of state that is not necessarily the same as the one of the underlying star. For example, one could impose that the Lagrangian perturbation of the pressure of the star always vanishes. This would correspond to a dust-like perturbation. However, in this case, the equations describe the evolution of the dust particles on a background describing a static star, rather than the motion of the star itself. Hence, one cannot conclude anything about stability in such an approach, but rather only about the behavior of the new matter in the star. For example, one can conclude that grains of dust in the Sun move either toward the surface or the center, not that the Sun is unstable.

With all of these considerations, the remaining mathematical difficulty in the problem is that both the coefficients \(p\) and \(w\) in the Sturm--Liouville problem change signs. In particular, the fact that \(w\) changes sign means we are not considering an eigenvalue problem in a Hilbert space, but rather in a more general vector space with an indefinite inner product. This keeps us from using standard results from functional analysis on Hilbert spaces. However, we can manipulate the Chandrasekhar pulsation equation into a form corresponding to a self-adjoint problem in a Hilbert space, as we explain below. The main observation is that all coefficients flip signs precisely when \(P+\rho\) vanishes (\(p\) and \(q\) can undergo additional sign flips, but these are unimportant). As a consequence, the problem can be multiplied by the sign of \(P+\rho\) to yield a sufficiently regular problem. This allows the use of some of the techniques in Ref. \cite{bardeen1966CatalogueMethodsStudying}. The main available technique is method 2-D, which is numerically costly, but manageable for our purposes. 

\subsection{Alternative Form of the Chandrasekhar Pulsation Equation}
    We assume the perturbations to be adiabatic, so that Eq. (\ref{eq: gamma-adiabatic}) holds. Consider the Chandrasekhar equation in the Sturm--Liouville form of Eq. (\ref{eq: SLP}), with the relevant functions given in Eqs. (\ref{eq: SLP-p-original}) to (\ref{eq: SLP-w-original}) and (\ref{eq: gamma-adiabatic}). Since the TOV equation establishes that \(\dv*{P}{r}\) is proportional to \(P+\rho\), we see that \(p(r)\), \(q(r)\), and \(w(r)\) are also proportional to \(P+\rho\). This combination can flip its sign inside the star due to NEC violations and cause \(w\) to reverse its sign as well. Suppose we eliminate this sign term from the equation. In that case, we will get to a more standard Sturm--Liouville problem, which in the worst-case scenario will still correspond to a self-adjoint eigenvalue problem in a Hilbert space (rather than an indefinite inner product space).

    This can be easily done by defining the function 
    \begin{equation}
        s(r) = \mathrm{sign}(P(r)+\rho(r)).
    \end{equation}
    Where the sign function is defined by
    \begin{equation}
        \mathrm{sign}(x) = \begin{cases}
            +1, & \text{if } x > 0, \\
            0, & \text{if } x = 0, \\
            -1, & \text{if } x < 0.
        \end{cases}
    \end{equation}

    We can now write the Chandrasekhar pulsation equation as
    \begin{equation}\label{eq: modified-SLP}
        \dv{r}\qty[s(r)p(r) \dv{\zeta}{r}] + s(r)q(r) \zeta(r) + \sigma^2 s(r)w(r) \zeta(r) = 0,
    \end{equation}
    which is possible because \(p\), \(q\), and \(w\) all vanish at the point in which \(s(r)\) does. In particular, the derivative of \(s(r)\) does not lead to any problems because \(p(r)\) vanishes at the point in which the derivative does not. By inspecting Eq. (\ref{eq: SLP-w-original}) one can tell that \(s(r)w(r)\) is manifestly non-negative, and this ensures the Sturm--Liouville problem takes place in a Hilbert space. 

    For future convenience, we introduce the shorthand notation
    \begin{equation}
        p_{\star}(r) = s(r)p(r)
    \end{equation}
    and similarly for \(q_{\star}\) and \(w_{\star}\).

\subsection{Variational Technique for Stability Analysis}
    To see whether a Sturm--Liouville problem admits negative eigenvalues, we can use the so-called min-max principle (Theorem XIII.1 in Ref. \cite{reed1978AnalysisOperators}). For our purposes, it states that given a self-adjoint operator \(L\) with spectrum \(\mathrm{spec}(L)\) on a Hilbert space, it holds that
    \begin{equation}
        \inf \mathrm{spec}(L) \leq \frac{\braket{\psi}{L\psi}}{\braket{\psi}}
    \end{equation}
    for all \(\psi \in \mathrm{Dom}(L)\) with \(\psi \neq 0\). Ref. \cite{reed1978AnalysisOperators} makes the additional assumption that \(L\) is bounded from below, but since we are only interested in knowing whether \(\mathrm{spec} (L)\) extends to negative values, we do not need this additional assumption (having a spectrum which extends to minus infinity means there is a negative value in the spectrum). This is a costly albeit standard technique in the study of stellar stability used by Chandrasekhar in the same references in which the pulsation equation was introduced \cite{chandrasekhar1964DynamicalInstabilityApJ,chandrasekhar1964DynamicalInstabilityPRL}. It corresponds to method 2-D in Ref. \cite{bardeen1966CatalogueMethodsStudying}. This method is also widely used in textbook quantum mechanics to find upper bounds on the ground state energy of a system \cite{sakurai2021ModernQuantumMechanics,weinberg2015LecturesQuantumMechanics}.
    
    In the Sturm--Liouville problem with coefficient functions \(p_{\star}\), \(q_{\star}\), and \(w_{\star}\) the inner product is given by 
    \begin{equation}
        \braket{\varphi}{\psi} = \int \varphi(r) \psi(r) w_{\star}(r) \dd{r},
    \end{equation}
    where we assume the functions to be real. We can then cast the problem in the form
    \begin{equation}
        L \zeta = \sigma^2 \zeta,
    \end{equation}
    with 
    \begin{equation}
        L = - \frac{1}{w_{\star}(r)}\qty[\dv{r}\qty[p_{\star}(r) \dv{r}] + q_{\star}(r)].
    \end{equation}
    
    \(L\) is self-adjoint \footnote{The standard references \cite{chandrasekhar1964DynamicalInstabilityApJ,bardeen1966CatalogueMethodsStudying} do not pay close attention to the difference between symmetric and self-adjoint operators, so neither will we. Details of this type within Sturm--Liouville theory can be found, for example, in Ref. \cite{pryce1993NumericalSolutionSturm}.} when we assume homogeneous boundary conditions (our case of interest). Since \(w_{\star}\) can only be zero at points in which \(p_{\star}\) and \(q_{\star}\) vanish (which turns out to be a point in which \(p_{\star}'\) also vanishes) and all of these functions vanish in precisely the same way, the operator \(L\) is well-defined.
    
    Using integration by parts one can write
    \begin{multline}
        \braket{\psi}{L\psi} = p_{\star}(0)\psi'(0)\psi(0) - p_{\star}(R)\psi'(R)\psi(R) \\ + \int_0^R p_{\star}(r) \psi'(r)^2 - q_{\star}(r)\psi(r)^2 \dd{r}.
    \end{multline}
    The boundary conditions we are interested in, given in Eq. (\ref{eq: boundary-SLP}), are such that we can assume the boundary contributions to vanish. Hence, we can write that 
    \begin{equation}\label{eq: S-psi}
        \inf \mathrm{spec}(L) \leq \frac{\int_0^R p_{\star}(r) \psi'(r)^2 - q_{\star}(r)\psi(r)^2 \dd{r}}{\int_0^R w_{\star}(r) \psi(r)^2 \dd{r}} \equiv S[\psi],
    \end{equation}
    where \(\psi\) is assumed to satisfy the boundary conditions in Eq. (\ref{eq: boundary-SLP}). Notice that the question of stability has been reduced to whether \(S[\psi]\) admits negative values for some \(\psi\) satisfying the boundary conditions.

\subsection{Bardeen's Technique for Stability Analysis}
    An alternative method of studying the stability of a star using the Chandrasekhar pulsation equation is due to Bardeen \cite{bardeen1965StabilityDynamicsSpherically} (see also Method 2-C in Ref. \cite{bardeen1966CatalogueMethodsStudying}). It uses the Sturm Comparison Theorem \cite{zettl2005SturmLiouvilleTheory,teschl2012OrdinaryDifferentialEquations} and the fact that the \(n\)-th eigenfunction of a Sturm--Liouville problem has \(n\) roots to establish the number of unstable modes in the star. The basic idea is that increasing the eigenvalue in the Sturm--Liouville equation increases the number of roots a solution with fixed initial (not boundary) conditions has in the interval (this follows from the Sturm Comparison Theorem). As a consequence, one can count how many roots the numerical solution with \(\sigma^2=0\) has, and this will correspond to the number of eigenvalues with \(\sigma^2 < 0\). This number therefore corresponds to the number of unstable modes in the star.

    The basic process is the following. One numerically solves the pulsation equation with the ansatz \(\sigma^2 = 0\). If the stellar surface is a regular point for the Sturm--Liouville problem, then one starts from there. Otherwise, one starts from the center. One counts how many roots the numerical solution has inside the star and this gives the number of unstable modes. If the numerical solution reaches the other end of the star (either the center or the boundary, depending on where integration started) and satisfies the boundary condition there, then \(\sigma^2 = 0\) is one of the eigenvalues of the problem, meaning there is a mode with neutral stability.

    This method is much more efficient than the variational technique but has some drawbacks. The variational technique relies on a very general result from functional analysis (the min-max principle), while Bardeen's method depends on the Sturm Comparison Theorem, which is a result about a certain class of Sturm--Liouville problems. In particular, the theorem requires that \(p_{\star}\) be positive on \((0,R)\), and hence Bardeen's technique cannot be applied to stars in which \(P+\rho\) flips signs or in which \(\qty(\pdv*{P}{\rho})_s < 0\). This excludes all acceptable models of negative-mass stars (see section \ref{sec: stability-negative-mass}), but the method is still convenient to illustrate how the variational technique works.

\subsection{Example: Relativistic Polytropes}
    As a warm-up, let us use these stability criteria against some simple examples. Namely, we consider relativistic polytropes in the sense originally considered by Tooper \cite{tooper1964GeneralRelativisticPolytropic}. These models were also studied by Chandrasekhar in Ref. \cite{chandrasekhar1964DynamicalInstabilityApJ} as an example application of the original pulsation equation.

    The equation of state we will consider is
    \begin{equation}\label{eq: tooper-polytrope}
        P = K\rho^\Gamma,
    \end{equation}
    where \(\Gamma\) and \(K\) are constants. \(K\) is a dimensional constant, but since it does not affect our results (except for appropriate rescaling of functions), we set \(K=1\) for our numerical computations. 
    
    The mass-radius diagrams and mass per central density diagrams for some values of \(\Gamma\) are shown in figures \ref{fig: tooper-polytrope-2}, \ref{fig: tooper-polytrope-53}, and \ref{fig: tooper-polytrope-43}. In the Newtonian case, \(\Gamma = \frac{4}{3}\) and \(\Gamma = \frac{5}{3}\) are common stellar models \cite{weinberg2020LecturesAstrophysics}.

    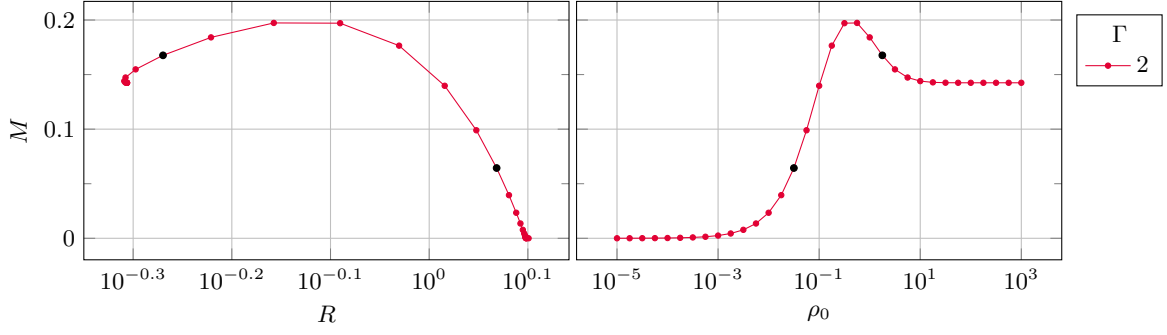
\begin{figure*}
        \centering
        \begin{tikzpicture}
            \begin{groupplot}[
                group style={group size=2 by 1,y descriptions at=edge left,horizontal sep=0.1cm,},
                minor tick num = 1,
                xmajorgrids=true,
                ymajorgrids=true,
                ylabel = {\(M\)},
                legend style={at={(1.03,0.95)},anchor=north west},
                width = 8cm,
                height= 5cm
            ]
            \nextgroupplot[xlabel = {\(R\)},
                xmode = log
                ]
                
                \addplot[NickRed,mark=*,mark size=1pt] table [x=R,y=M] {datapos63.dat};

                \addplot[black,only marks, mark=*,mark size=1.25pt] table [x=R,y=M] {datapos63select.dat};
                
            \nextgroupplot[xlabel = {\(\rho_0\)},xmode=log,
                xtick = {1e-5,1e-3,1e-1,1e1,1e3},
                ]
            
                \addlegendimage{empty legend}
                    \addlegendentry{\hspace{-.6cm}\(\Gamma\)}
                
                \addplot[NickRed,mark=*,mark size=1pt] table [x=rhoc,y=M] {datapos63.dat};
                    \addlegendentry{\(2\)}

                \addplot[black,only marks, mark=*,mark size=1.25pt] table [x=rhoc,y=M] {datapos63select.dat};
            \end{groupplot}
        \end{tikzpicture}
        \caption{Mass-radius diagram and mass per central density \(\rho_0\) for stars satisfying the equation of state (\ref{eq: tooper-polytrope}) with \(\Gamma = 2\). The points indicated in black will be used to perform the detailed stability analysis.}
        \label{fig: tooper-polytrope-2}
    \end{figure*}

    \begin{figure*}
        \centering
        \begin{tikzpicture}
            \begin{groupplot}[
                group style={group size=2 by 1,y descriptions at=edge left,horizontal sep=0.1cm,},
                minor tick num = 1,
                xmajorgrids=true,
                ymajorgrids=true,
                ylabel = {\(M\)},
                legend style={at={(1.03,0.95)},anchor=north west},
                width = 8cm,
                height= 5cm
            ]
            \nextgroupplot[xlabel = {\(R\)},
                xmode = log
                ]
                
                \addplot[NickGreen,mark=triangle*,mark size=1.2pt] table [x=R,y=M] {datapos53.dat};

                \addplot[black,only marks, mark=triangle*,mark size=1.5pt] table [x=R,y=M] {datapos53select.dat};
            \nextgroupplot[xlabel = {\(\rho_0\)},xmode=log,
                ]
            
                \addlegendimage{empty legend}
                    \addlegendentry{\hspace{-.6cm}\(\Gamma\)}
                
                \addplot[NickGreen,mark=triangle*,mark size=1.2pt] table [x=rhoc,y=M] {datapos53.dat};
                    \addlegendentry{\(\frac{5}{3}\)}

                \addplot[black,only marks, mark=triangle*,mark size=1.5pt] table [x=rhoc,y=M] {datapos53select.dat};
            \end{groupplot}
        \end{tikzpicture}
        \caption{Mass-radius diagram and mass per central density \(\rho_0\) for stars satisfying the equation of state (\ref{eq: tooper-polytrope}) with \(\Gamma = \frac{5}{3}\). The points indicated in black will be used to perform the detailed stability analysis.}
        \label{fig: tooper-polytrope-53}
    \end{figure*}
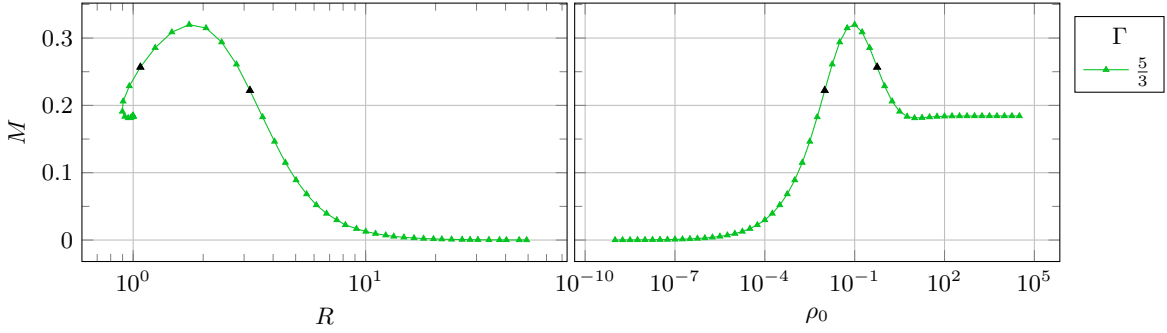

    \begin{figure*}
        \centering
        \begin{tikzpicture}
            \begin{groupplot}[
                group style={group size=2 by 1,y descriptions at=edge left,horizontal sep=0.1cm,},
                minor tick num = 1,
                xmajorgrids=true,
                ymajorgrids=true,
                ylabel = {\(M\)},
                legend style={at={(1.03,0.95)},anchor=north west},
                width = 8cm,
                height= 5cm
            ]
            \nextgroupplot[xlabel = {\(R\)},
                xmode = log
                ]
                
                \addplot[NickYellow,mark=square*,mark size=1pt] table [x=R,y=M] {datapos43.dat};

                \addplot[black,only marks, mark=square*,mark size=1.25pt] table [x=R,y=M] {datapos43select.dat};
                
            \nextgroupplot[xlabel = {\(\rho_0\)},xmode=log,
                ]
            
                \addlegendimage{empty legend}
                    \addlegendentry{\hspace{-.6cm}\(\Gamma\)}
                
                \addplot[NickYellow,mark=square*,mark size=1pt] table [x=rhoc,y=M] {datapos43.dat};
                    \addlegendentry{\(\frac{4}{3}\)}

                \addplot[black,only marks, mark=square*,mark size=1.25pt] table [x=rhoc,y=M] {datapos43select.dat};
            \end{groupplot}
        \end{tikzpicture}
        \caption{Mass-radius diagram and mass per central density \(\rho_0\) for stars satisfying the equation of state (\ref{eq: tooper-polytrope}) with \(\Gamma = \frac{4}{3}\). The points indicated in black will be used to perform the detailed stability analysis.}
        \label{fig: tooper-polytrope-43}
    \end{figure*}
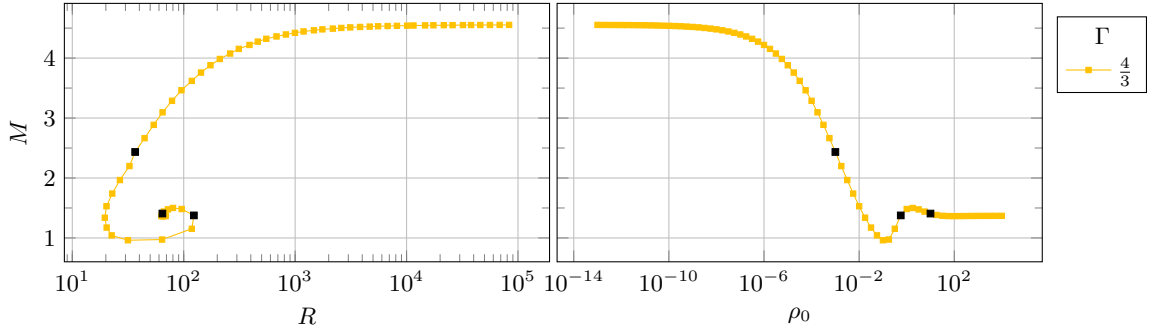

    The spiral structures in the mass-radius diagrams for these three different sequences of stars are a very generic feature of the TOV equation, as per Theorem 6.4 in Ref. \cite{heinzle2003DynamicalSystemsApproach}. The most important feature for stability analyses is the sign of 
    \begin{equation}
        \pdv{M}{\rho_0},
        \label{elementarycrit}
    \end{equation}
    where \(M\) is the equilibrium value of the mass, for it provides an elementary stability criterion (see, \eg, Ref. \cite{glendenning1997CompactStarsNuclear}). Indeed, consider a sequence of stars and pick a star with mass \(\bar M > 0\) and \(\pdv{M}{\rho_0} > 0\). Perturb this star by increasing the central density according to
    \begin{equation}
        \rho_0 \to \rho_0 + \var{\rho_0}
    \end{equation}
    (supposing \(\rho_0, \var{\rho_0} > 0\)) assuming that \(\bar M\) is kept constant during this process. Since \(\pdv*{M}{\rho_0} > 0\), \(\bar M\) is smaller than the mass \(M\) of the equilibrium star with central density \(\rho_0 + \var{\rho_0}\). This means the gravitational field of the perturbed star is weaker than it should be to retain equilibrium. Consequently, the increased pressure forces the star to expand, thereby decreasing the central pressure and density, and bringing the perturbed star back to the original equilibrium state. Notice that if we had \(\pdv*{M}{\rho_0} < 0\) instead, the perturbation would grow and the star would be unstable. It should be pointed out that this is a necessary, but not sufficient, criterion for stability.

    To perform more detailed stability analyses, we need to fix the boundary conditions to be considered in the Chandrasekhar pulsation equation. Following Eq. (\ref{eq: boundary-SLP}), we take 
    \begin{equation}\label{eq: boundary-stability-0}
        \zeta(r) \sim r^3 \text{ at } r \approx 0
    \end{equation}
    and
    \begin{equation}\label{eq: boundary-stability-R}
        \dv{\zeta}{r} = 0 \text{ at } r = R.
    \end{equation}
    
    We will first use the variational technique to obtain the sign of the lowest oscillation mode. We pick a trial function and compute \(S[\psi]\) to check its sign. If it is negative, the star is unstable. To abide by the boundary conditions, we will use trial functions of the form 
    \begin{equation}\label{eq: trial-function}
        \psi_n(r) = r^3 - \frac{3 R^{-n}}{3+n} r^{3+n}
    \end{equation}
    where the coefficient was adjusted so that \(\psi_n'(R) = 0\). The constant \(n\) must be positive to obey the boundary conditions, but it is otherwise arbitrary. Such trial functions were considered in Ref. \cite{meltzer1966NormalModesRadial} with more elaborate techniques to estimate the first few eigenvalues of the Chandrasekhar pulsation equation. In scenarios in which the star is very inhomogeneous these trial functions can lead to misleading estimates. Yet, since we are only interested in the sign of the smallest eigenvalue we do not need to worry about precise approximations.

    Using these trial functions, we can compute \(S[\psi_n]\) and check its sign. The results are shown in table \ref{tab: stability-tooper-polytrope}. We see that the variational criterion concurs with the expectations based on the sign of Eq.~\eqref{elementarycrit} for the sequences with \(\Gamma = 2\) and \(\Gamma = \frac{5}{3}\), but not for \(\Gamma = \frac{4}{3}\). The sequence with \(\Gamma = \frac{4}{3}\), the star with \(\rho_0 = \num{5.623e-1}\) satisfies $\pdv{M}{\rho_0}>0$ but is not stable. In fact, it is well known \cite{chandrasekhar1964DynamicalInstabilityApJ,glendenning1997CompactStarsNuclear} that in Newtonian theory a polytrope is only stable if \(\Gamma > \frac{4}{3}\), and the bound gets tighter in relativity. This is consistent with the fact that the positivity of~\eqref{elementarycrit} is only a necessary, not sufficient, condition for stability.
    
    \begin{table}[t]
        \centering 
        \caption{Values of \(n\) for the trial function of Eq. (\ref{eq: trial-function}) and of the quantity \(S[\psi_n]\)  defined on Eqs. (\ref{eq: S-psi}) for stars with equation of state of the form (\ref{eq: tooper-polytrope}) and different central densities \(\rho_0\). The different choices of central densities can be identified as the black points in figures \ref{fig: tooper-polytrope-2} to \ref{fig: tooper-polytrope-43}.}
        \label{tab: stability-tooper-polytrope}
        \begin{ruledtabular}
        \begin{tabular}{cccc}
            \(\Gamma\) & \(\rho_0\) & \(n\) & \(S[\psi_n]\) \\ \hline
            \num{2} & \num{3.162e-2} & \num{2} & \num{2.065e-4} \\ 
            \num{2} & \num{1.778} & \num{2} & \num{-1.408e-3} \\
            \(\num{5}/\num{3}\) & \num{1.000e-2} & \num{2} & \num{4.793e-4} \\ 
            \(\num{5}/\num{3}\) & \num{5.623e-1} & \num{2} & \num{-2.046e-3} \\
            \(\num{4}/\num{3}\) & \num{1.000e-3} & \num{2} & \num{-7.337e-3} \\ 
            \(\num{4}/\num{3}\) & \num{5.623e-1} & \num{2} & \num{-7.340e-3} \\
            \(\num{4}/\num{3}\) & \num{1.000e1} & \num{2} & \num{-7.242e-3}
        \end{tabular}
        \end{ruledtabular}
    \end{table}

    Of course, the positive values of \(S[\psi_n]\) in table \ref{tab: stability-tooper-polytrope} do not imply stability, they simply fail to falsify it. In the absence of other methods, we could repeat the computation for other values of \(n\) and choices of test functions.

    In these examples, the pulsation equation is sufficiently well-behaved to employ Bardeen's technique. We report the results in table \ref{tab: bardeen-stability-tooper-polytrope}, all of which agree with the results obtained through variational methods. Notice that Bardeen's technique can establish stability and confirms that the positive values of table \ref{tab: stability-tooper-polytrope} indeed correspond to stable stars.

    \begin{table}[t]
        \centering
        \caption{Number of unstable modes according to Bardeen's method \cite{bardeen1965StabilityDynamicsSpherically,bardeen1966CatalogueMethodsStudying} for stars with equation of state of the form (\ref{eq: tooper-polytrope}) and different central densities \(\rho_0\). The different choices of central densities can be identified as the black points in figures \ref{fig: tooper-polytrope-2} to \ref{fig: tooper-polytrope-43}.}
        \label{tab: bardeen-stability-tooper-polytrope}
        \begin{ruledtabular}
        \begin{tabular}{ccc}
            \(\Gamma\) & \(\rho_0\) & unstable modes \\ \hline
            \num{2} & \num{3.162e-2} & \num{0} \\ 
            \num{2} & \num{1.778} & \num{1} \\
            \(\num{5}/\num{3}\) & \num{1.000e-2} & \num{0} \\ 
            \(\num{5}/\num{3}\) & \num{5.623e-1} & \num{1} \\
            \(\num{4}/\num{3}\) & \num{1.000e-3} & \num{1} \\ 
            \(\num{4}/\num{3}\) & \num{5.623e-1} & \num{2} \\
            \(\num{4}/\num{3}\) & \num{1.000e1} & \num{3}
        \end{tabular}
        \end{ruledtabular}
    \end{table}

\section{Examples of Negative-Mass Stars}\label{sec: examples}
We are now ready to show different stellar models with negative masses and study their stability. 

The two main methods for generating negative-mass stellar models are explored in Ref. \cite{novikov2018StarsCreatingGravitational}: proposing an equation of state or specifying an energy density profile. The latter is more adequate for generating unusual stars that partially satisfy the ANEC since it allows for the fine-tuning of the stellar energy. We focus on this method first. 

\subsection{Solutions with Density Profiles}
    Generating solutions with a density profile involves different techniques from solving the TOV system with an equation of state. With an equation of state one integrates from the center outward to ensure the condition that \(m(0) = 0\), which is necessary to avoid a curvature singularity at \(r=0\). With a density profile, however, this condition is already ensured because one can define \(m\) directly as 
    \begin{equation}
        m(r) = \int_0^r 4 \pi \rho(r') r'^2 \dd{r'},
    \end{equation}
    where \(\rho(r)\) is a previously chosen function. Hence, it is secured that \(m(0) = 0\). Therefore, the initial condition we use instead is that \(P(R) = 0\) for some previously chosen \(R\). For simplicity, we take \(R=1\). This immediately determines that the star's mass is \(M = m(1)\), which can be fine-tuned by carefully selecting \(\rho(r)\). We now solve the TOV system for the pressure, as the remaining variables \(\rho\) and \(m\) are already known. 
    
    It will also be interesting to consider the value of the quantity 
    \begin{equation}\label{eq: F-rho}
        F[\rho] = \int_0^{R} (P+\rho) \frac{e^{-\phi(r)}}{\sqrt{1 - \frac{2m(r)}{r}}} \dd{r}
    \end{equation}
    which, on account of Eq. (\ref{eq: radial-ANEC-integral-in-TOV}), is proportional to the value of the ANEC integral along a radial null geodesic. The sign of \(F[\rho]\) for a given density profile \(\rho\) decides whether the ANEC is satisfied for radial null geodesics. We know that, due to the Penrose--Sorkin--Woolgar theorem and Borde's theorem, the AANEC will fail in all of the negative-mass stellar spacetimes we are about to consider. However, \(F[\rho]\) allows us to understand this failure in more detail. 
    
    \begin{table}[t]
        \centering
        \caption{Mass \(M\) and radial ANEC integral \(F[\rho]\) (see Eq. (\ref{eq: F-rho})) for different profile choices \(\rho\). The leftmost column labels each stellar model to facilitate discussion in the main text.}
        \label{tab: profiles}
        \begin{ruledtabular}
        \begin{tabular}{Nccc}
            \multicolumn{1}{c}{model} & \(\rho(r)\) & \(M\) & \(F[\rho]\) \\ \hline
            \label{model:const} & \(-\frac{3}{4\pi}\) & \num{-1.000} & \num{-8.619e-2} \\
            \label{model:r-1} & \(r-1\) & \num{-1.047} & \num{-1.506e-1} \\
            \label{model:r-2} & \(r-2\) & \num{-5.236} & \num{-1.639e-1} \\
            \label{model:r2-1} & \(r^2-1\) & \num{-1.676} & \num{-1.511e-1} \\
            \label{model:r2-2} & \(r^2-2\) & \num{-5.864} & \num{-1.645e-1} \\
            \label{model:r3-1} & \(r^3-1\) & \num{-2.094} & \num{-1.485e-1} \\
            \label{model:r3-2} & \(r^3-2\) & \num{-6.283} & \num{-1.636e-1} \\
            \label{model:exp-e} & \(\exp(r)-e\) & \num{-2.360} & \num{-1.755e-1} \\
            \label{model:exp-2e} & \(\exp(r)-2e\) & \num{-1.375e1} & \num{-1.869e-1} \\
            \label{model:-cos} & \(-\cos(\frac{\pi r}{2})\) & \num{-1.515} & \num{-1.524e-1} \\
            \label{model:-r} & \(-r\) & \num{-3.142} & \num{-7.964e-2} \\
            \label{model:-r2} & \(-r^2\) & \num{-2.513} & \num{-6.085e-2} \\
            \label{model:-r3} & \(-r^3\) & \num{-2.094} & \num{-5.129e-2} \\
            \label{model:-exp} & \(-\exp(r)\) & \num{-9.026} & \num{-1.282e-1} \\
            \label{model:-sin}& \(-\sin(\frac{\pi r}{2})\) & \num{-3.701} & \num{-8.832e-2} \\
            \label{model:12-r}& \(\frac{1}{2}-r\) & \num{-1.047} & \num{5.453e-2} \\
            \label{model:25-r2}& \(\frac{2}{5}-r^2\) & \num{-8.378e-1} & \num{2.039e-1} \\
            \label{model:12cos} & \(\frac{1}{2}\cos(\pi r)\) & \num{-1.273} & \num{9.402e-2}
        \end{tabular}
        \end{ruledtabular}
    \end{table}
    
    Table \ref{tab: profiles} exhibits the masses and values of \(F[\rho]\) for different profile choices. All models have positive pressures in the interior of the star. Models \ref{model:const}, \ref{model:r2-1}, and \ref{model:r2-2} are particular cases of the ones considered in Ref. \cite{novikov2018StarsCreatingGravitational}. Models \ref{model:r-1}, \ref{model:r-2}, and \ref{model:r3-1} to \ref{model:-cos} are natural generalizations of those models. All of these models have a large amount of negative energy density in the center of the star, and the absolute value of the energy density decreases from the center to the border. This is similar to a regular star, apart, of course, from the sign. Models \ref{model:const} to \ref{model:exp-2e} consistently violate the NEC in the sense that \(P(r) + \rho(r) < 0\) for all \(r < R\) (some of them have \(P(R)+\rho(R) = 0\)). In this sense, they are completely made of negative energy.

    Models \ref{model:-r} to \ref{model:-sin} are less usual. They start with a small amount of negative energy density at the center (often taken to be zero for simplicity) and the absolute value of the energy density grows as one gets farther away. This is an uncommon situation with the peculiar feature that the equation of state yields a non-vanishing pressure for a vanishing energy density. Still, we include this case for completeness. These models satisfy the NEC in the deep interior of the star because the pressure is positive, but the energy density is very small. The NEC is still violated on the outer layers (as the Penrose--Sorkin--Woolgar theorem demands).

    Models \ref{model:12-r} to \ref{model:12cos} are the most interesting examples. Their density profiles are plotted in figure \ref{fig: jk-profiles}. These density profiles flip signs inside the star. In these two particular cases, the stars have a positive mass core, but their outer layers are made of negative mass densities. Since the pressure is everywhere positive inside the star, the core satisfies the NEC, but the outer layers do not. As one can tell from table \ref{tab: profiles}, this can be fine-tuned in such a way that the ANEC is respected by radial geodesics (even though the Penrose--Sorkin--Woolgar theorem implies it must be violated for other geodesics). The reason is that radial null geodesics will cross a region with a substantial amount of positive energy, which suffices to balance out the negative amounts obtained in the outer shells. Nevertheless, since the outer shells have negative energy densities, a ``glancing'' null geodesic, \ie, one that barely penetrates the stellar surface, can cross a region of strictly negative energy and violate the ANEC. This is illustrated in figure \ref{fig: geodesics}. Incidentally, the Penrose--Sorkin--Woolgar theorem entails that at least one of these glancing geodesics must be inextendible and achronal.
    
    \begin{figure}[t]
        \centering
        \begin{tikzpicture}
            \begin{axis}[
                ymin=-0.6,ymax=0.6,
                xmin=0,xmax=1,
                minor tick num = 1,
                xmajorgrids=true,
                ymajorgrids=true,
                legend style={at={(0.97,0.95)},anchor=north east},
                xlabel = \(r\),
                ylabel = {\(\rho(r)\)},
                width = 8cm,
                height= 5cm
            ]
            
            \addplot [
                domain=0:1, 
                samples=100, 
                color=NickRed,
                line width = 1.2pt,
                ]
                {0.5-x};
            \addlegendentry{\ref*{model:12-r}}
            
            \addplot [
                domain=0:1, 
                samples=100, 
                color=NickGreen,
                line width = 1.2pt,
                ]
                {0.2-x*x};
            \addlegendentry{\ref*{model:25-r2}}

            \addplot [
                domain=0:1, 
                samples=100, 
                color=NickYellow,
                line width = 1.2pt,
                ]
                {0.5*cos(180*x)};
            \addlegendentry{\ref*{model:12cos}}
            \end{axis}
        \end{tikzpicture}
        \caption{Energy density profiles of models \ref{model:12-r} to \ref{model:12cos} in table \ref{tab: profiles}.}
        \label{fig: jk-profiles}
    \end{figure}
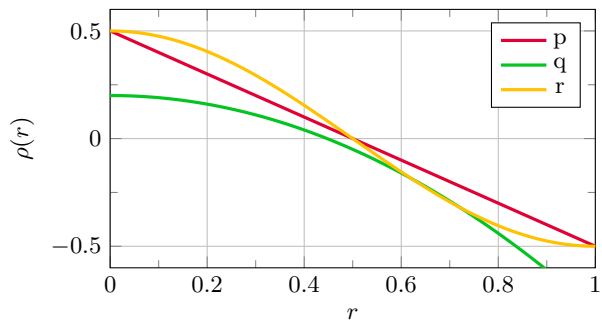

    \begin{figure}[t]
        \centering
        \begin{tikzpicture}
            \begin{axis}[
                trig format plots=rad,
                axis equal,
                hide axis,
                width = 8cm,
                height= 8cm
            ]
                \addplot [domain=0:2*pi, samples=200, black] ({cos(x)}, {sin(x)});
                
                \addplot[NickRed] table [col sep=comma] {l025.csv};
                \addplot[NickRed] table [col sep=comma] {l050.csv};
                \addplot[NickRed] table [col sep=comma] {l075.csv};
            \end{axis}
        \end{tikzpicture}
        \caption{Different null geodesics in the spacetime of star model \ref{model:12-r} (see table \ref{tab: profiles}). The black circle represents the star's surface. The geodesic that comes the closest to the center has a positive value for the ANEC integral, while the remaining ones violate the ANEC.}
        \label{fig: geodesics}
    \end{figure}
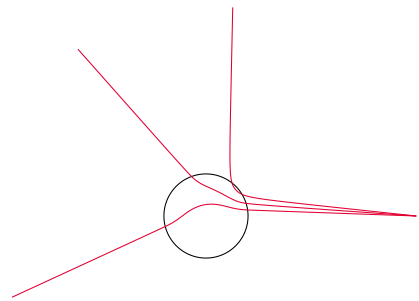
    
    Once we consider the gravitational dynamics of particles with different mass signs as discussed in the introduction, configurations with sign flips appear to be unstable. The positive core is repelled by the outer layers but has nowhere to go. The perfect spherical symmetry of the star keeps the core in what seems to be an unstable equilibrium. Meanwhile, the outer layers fall toward the core. This argument suggests the impossibility of the reversed situation: a negative core would be attracted by the positive layers, which would be repelled by the core. It seems one would need, at the very least, negative pressures in the positive-energy region to keep the system together. Indeed, suppose the total mass \(M\) is negative and that \(\rho(R) > 0\). Since \(P(R) = 0\), we see that the TOV equation leads to
    \begin{equation}
        \eval{\dv{P}{r}}_{R} = - \rho(R) \frac{M}{R(R-2M)} > 0.
    \end{equation}
    Hence, \(P(r)\) must be negative for \(r < R\) sufficiently close to \(R\). Negative pressures have been considered elsewhere in the study of ultra-compact objects, such as in the Mazur--Mottola gravastar \cite{mazur2004GravitationalVacuumCondensate,mazur2023GravitationalCondensateStars}. We shall not focus on this direction in this work due to the reasons discussed at the end of Section \ref{sec: stability-negative-mass}. 
    
    The above argument suggests that negative-mass stars with a sign flip should be unstable against perturbations that break spherical symmetry. It should be mentioned, however, that negative-mass stars without any sign flips also have curious equilibrium dynamics \cite{novikov2018StarsCreatingGravitational}. Since the gravitational interaction between negative masses is repulsive, the gravitational force tries to make a negative-mass star explode. Meanwhile, the pressure gradient applies an outward-pointing force on the fluid elements, which are then accelerated inward due to the negative mass sign. Hence, in a negative-mass star, the pressure pulls the star inward while gravity tries to make the star explode (rather than collapse). Thus, the roles of pressure and gravity in the stellar equilibrium dynamics are reversed.

\subsection{Solutions with Equations of State}
    Next, we consider examples of stars with a negative mass arising from a (barotropic) equation of state. 

    First, we should mention it seems to be particularly difficult to use this method to generate finite stars with a sign flip, such as models \ref{model:12-r} and \ref{model:25-r2} of table \ref{tab: profiles}. The reason is as follows. Assuming the star has positive pressure, there is a point inside the star where the NEC holds (because the pressure is always positive and the sign flip implies there is some point with positive energy density). If the NEC held everywhere, the star would have positive mass due to the Penrose--Sorkin--Woolgar theorem. Hence, the NEC must be violated somewhere. By continuity, this implies that there is a point with \(P+\rho = 0\). Since the pressure obeys the TOV equation and we are assuming the equation of state to be barotropic, we find that
    \begin{equation}\label{eq: rho-tov}
        \dv{\rho}{r} = - (P+\rho) \frac{4\pi P r^3 + m}{r[r-2m]} \dv{\rho}{P}.
    \end{equation}
    We see then that, at the point with \(P+\rho = 0\), both the derivatives of \(P\) and \(\rho\) typically vanish on account of Eqs. (\ref{eq: tov}) and (\ref{eq: rho-tov}). Since the differential equations are of first order, this means \(P\) and \(\rho\) become constant from that point onward, which implies the star will be infinite.

    This behavior can be bypassed if the equation of state is fine-tuned so that \(\dv*{\rho}{P}\) diverges precisely when \(P+\rho\) vanishes. This makes \(P+\rho\) traverse the point where it is zero, and the TOV equation allows the pressure to continue evolving. 

    Such fine-tuning can be easily enforced when creating an equation of state, so it is not problematic as long as one is aware of it. There is, however, a second concern. Both the TOV equation and Eq. (\ref{eq: rho-tov}) also have a factor of the form
    \begin{equation}
        4 \pi P r^3 + m.
    \end{equation}
    Assuming the NEC holds at the center of the star and that the pressure is positive, one can conclude that the strong energy condition (SEC) holds at the center of the star: \(\rho(0) + 3P(0) \geq 0\) and \(\rho(0) + P(0) \geq 0\) \cite{curiel2017PrimerEnergyConditions,kontou2020EnergyConditionsGeneral,martin-moruno2017ClassicalSemiclassicalEnergy,witten2020LightRaysSingularities}. In fact, we have \(\rho(0) + 3P(0) > 0\), since \(P(0) > 0\). This implies that, for sufficiently small \(r\),
    \begin{equation}
            4 \pi P r^3 + m = \frac{4\pi}{3} \int_0^r (3P(0) + \rho(r')) r'^2 \dd{r'} > 0,
    \end{equation}
    which follows from the fact that \(\rho(0) + 3P(0) > 0\) implies \(\rho(r) + 3P(0) > 0\) for sufficiently small \(r\), by continuity. Hence, \(4 \pi P r^3 + m\) is positive near the star's center. Meanwhile, it is negative at its boundary because there we have \(P(R) = 0\) and \(m(R) < 0\). Hence, it must flip sign at least once due to the intermediate value theorem. At the point in which \(4 \pi P r^3 + m = 0\), we run into a new problem as before: the evolution of \(P\) and \(\rho\) reaches an extremum. Since the derivative of \(m\) does not vanish, the solution manages to traverse this point and keep evolving. However, unless \(\dv*{\rho}{P}\) diverges to keep \(\dv*{\rho}{r}\) nonzero, \(\rho\) will ``turn around'': the extremum point \(r_0\) is reached and, if \(\rho\) decreased for \(r<r_0\), it starts to grow for \(r>r_0\). This means that if \(\rho\) approaches the value \(\rho^*\) in which \(P(\rho^*) = 0\) for \(r<r_0\) (signaling the end of the star), it departs from \(\rho^*\) for \(r>r_0\).

    Counteracting this effect requires an even finer tuning of the equation of state to ensure that \(\dv*{\rho}{P}\) diverges exactly at the point in which \(4 \pi P r^3 + m = 0\). This difficulty is illustrated for model \ref{model:12-r} of table \ref{tab: profiles} in figures \ref{fig: fit-eos} and \ref{fig: profile-vs-eos}.

    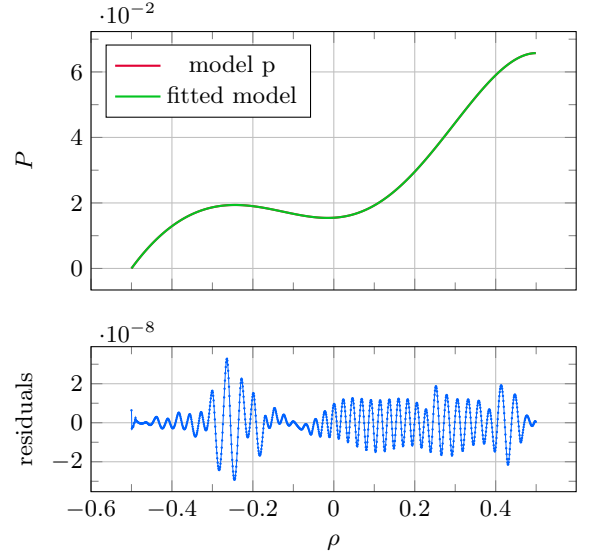
\begin{figure}[t]
        \centering
        \begin{tikzpicture}
            \begin{groupplot}[
                group style={group size=1 by 2,
                    x descriptions at=edge bottom,
                    vertical sep=0.75cm,
                    },
                minor tick num = 1,
                xmajorgrids=true,
                ymajorgrids=true,
                xlabel = {\(\rho\)},
                yticklabel style = {text width=1.35em, align=right},
                legend style={at={(0.03,0.95)},anchor=north west},
                width = 8cm,
            ]
            \nextgroupplot[ylabel = {\(P\)},
                height= 5cm
                ]
                
                \addplot[NickRed, thick] table [col sep=comma] {RawData.csv};
                    \addlegendentry{model \ref*{model:12-r}}
                \addplot[NickGreen, thick] table [col sep=comma] {FitData.csv};
                    \addlegendentry{fitted model}
            \nextgroupplot[height= 3.5cm,
                ylabel = {residuals},
                ]
                
                \addplot[NickBlue,mark=*,mark size=0.25pt,line width=0.05pt] table [col sep=comma] {FitError.csv};
            \end{groupplot}
        \end{tikzpicture}
        \caption{Fitted equation of state for the star model \ref{model:12-r} of table \ref{tab: profiles}. The fit was obtained by sampling a thousand points inside the star with evenly spaced values of \(r\). Using these thousand datapoints, we fitted a 25th-degree polynomial with the extra condition that \(\dv*{P}{\rho} = 0\) at the point with \(P(\rho) + \rho = 0\). Top: parametric plot of the equation of state obtained from the original profile by varying \(r\) along the star, together with the fitted equation of state. The difference between the two graphs is too small to be perceived. Bottom: residuals in the fit, obtained by subtracting the fit from the original values obtained directly from model \ref{model:12-r}.}
        \label{fig: fit-eos}
    \end{figure}

    \begin{figure}[t]
        \centering
        \begin{tikzpicture}
            \begin{groupplot}[
                group style={group size=1 by 2,
                    x descriptions at=edge bottom,
                    vertical sep=0.25cm,
                    },
                minor tick num = 1,
                xmajorgrids=true,
                ymajorgrids=true,
                xlabel = {\(r\)},
                yticklabel style = {text width=2.35em, align=right},
                legend style={at={(0.97,0.95)},anchor=north east},
                width = 8cm,
                height= 5cm
            ]
            \nextgroupplot[ylabel = {\(P\)},
                ]
                
                \addplot[NickRed,thick] table [x=r,y=P,col sep=comma] {ProfileSol.csv};
                    \addlegendentry{model \ref*{model:12-r}}
                \addplot[NickGreen,thick] table [x=r,y=P,col sep=comma] {EoSSol.csv};
                    \addlegendentry{fitted model}
            \nextgroupplot[
                ylabel = {\(\rho\)},
                ]
                
                \addplot[NickRed,thick] table [x=r,y=rho,col sep=comma] {ProfileSol.csv};
                \addplot[NickGreen,thick] table [x=r,y=rho,col sep=comma] {EoSSol.csv};
            \end{groupplot}
        \end{tikzpicture}
        \caption{Pressure and energy density for a star with the energy density profile of star model \ref{model:12-r} (see table \ref{tab: profiles}) and with the fitted equation of state of figure \ref{fig: fit-eos}. Notice that the tiny errors in the fit are sufficient to keep the star from behaving as desired. Those errors force the energy density to reach a local minimum, causing the pressure to depart from the point at which it would vanish. The star becomes infinite.}
        \label{fig: profile-vs-eos}
    \end{figure}
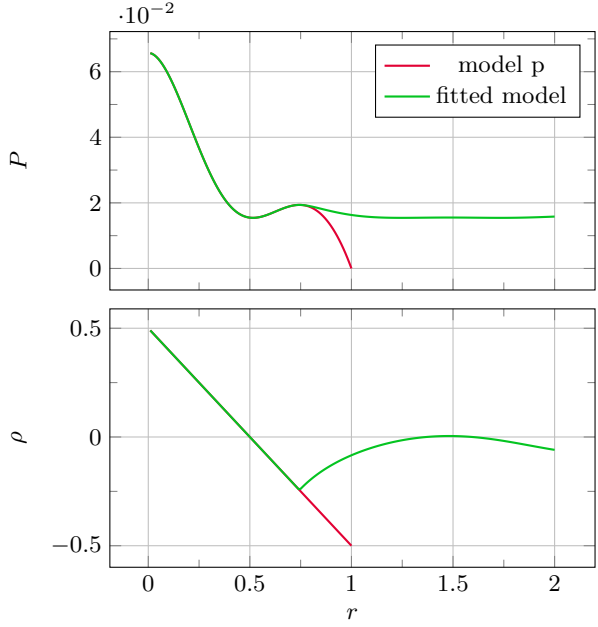 

    With this in mind, we focus on stars with negative-semidefinite energy density. Novikov, Bisnovatyi-Kogan, and Novikov previously considered two examples of this kind \cite{novikov2018StarsCreatingGravitational}, so we consider one of their models (their other model, with a linear equation of state, does not lead to finite-radius stars) and a few generalizations. More specifically, we will consider stars with an equation of state with the form
    \begin{equation}\label{eq: novikov-polytrope}
        P = (-1)^\Gamma \rho^\Gamma,
    \end{equation}
    where \(\Gamma\) is a positive integer for simplicity \footnote{In astrophysics, it is more common to drop the signal in Eq. (\ref{eq: novikov-polytrope}) and consider \(\Gamma = 1 + \frac{1}{n}\). This would be a relativistic polytrope, such as the ones considered in Section \ref{sec: stability}. However, since \(\rho\) is now negative, we keep \(\Gamma\) an integer so that the equation of state is differentiable.}. Ref. \cite{novikov2018StarsCreatingGravitational} considered this equation of state with \(\Gamma=1\)---which leads to an infinite star---and with \(\Gamma=2\). Notice this model is essentially a relativistic polytrope with negative mass.

    As noted in Ref. \cite{novikov2018StarsCreatingGravitational}, there is a bound on the parameters of the equation of state for the star to have a finite size. If the (negative) central density \(\rho_0\) is such that
    \begin{equation}\label{eq: novikov-polytrope-rho-crit}
        \rho_0 \leq \rho_c \equiv - 3^{\frac{1}{1-\Gamma}},
    \end{equation}
    then the star is infinite. The reason is that for \(\rho_0 = \rho_c\) the star admits a solution with a constant density (and hence a constant pressure). Notice that the bound diverges for \(\Gamma=1\), and all central densities lead to infinite stars.

    Figure \ref{fig: novikov-polytrope-plots} exhibits the mass-radius diagrams and mass per central density diagrams for a few stars satisfying the equation of state (\ref{eq: novikov-polytrope}). It is easy to see how the masses (and thus the radii) diverge as \(\rho_0\) approaches the critical central density \(\rho_c = -3^{\frac{1}{1-\Gamma}}\).

    \begin{figure*}
        \centering
        \begin{tikzpicture}
            \begin{groupplot}[
                group style={group size=2 by 1,y descriptions at=edge left,horizontal sep=0.1cm,},
                minor x tick num = 1,
                minor y tick num = 0,
                xmajorgrids=true,
                ymajorgrids=true,
                ylabel = {\(M\)},
                ytick={-11,-7,-3,1,5},
                yticklabels={\num{-e11},\num{-e7},\num{-e3},\num{-e-1},\num{-e-5}},
                legend style={at={(1.03,0.95)},anchor=north west},
                width = 8cm,
                height= 5cm
            ]
            \nextgroupplot[xlabel = {\(R\)},
                xmode = log]
                
                \addplot[NickRed,mark=*,mark size=1pt] table[x=R,y=log]\tabletwo;
            
                \addplot[NickGreen,mark=triangle*,mark size=1.2pt] table[x=R,y=log]\tablethree;
            
                \addplot[NickYellow,mark=square*,mark size=1pt] table[x=R,y=log]\tablefour;
            
                \addplot[NickPurple,mark=pentagon*,mark size=1.2pt] table[x=R,y=log]\tablefive;
            
            \nextgroupplot[xlabel = {\(\rho_0\)},
                ]
            
                \addlegendimage{empty legend}
                    \addlegendentry{\hspace{-.6cm}\(\Gamma\)}
                
                \addplot[NickRed,mark=*,mark size=1pt] table[x=rhoc,y=log]\tabletwo;
                    \addlegendentry{\(2\)}
            
                \addplot[NickGreen,mark=triangle*,mark size=1.2pt] table[x=rhoc,y=log]\tablethree;
                    \addlegendentry{\(3\)}
            
                \addplot[NickYellow,mark=square*,mark size=1pt] table[x=rhoc,y=log]\tablefour;
                    \addlegendentry{\(4\)}
            
                \addplot[NickPurple,mark=pentagon*,mark size=1.2pt] table[x=rhoc,y=log]\tablefive;
                    \addlegendentry{\(5\)}
            \end{groupplot}
        \end{tikzpicture}
        \caption{Mass-radius diagram and mass per central density \(\rho_0\) for negative-mass stars satisfying the equation of state (\ref{eq: novikov-polytrope}). Notice how the masses (and hence the radii) diverge close to the limiting value \(\rho_c = - 3^{\frac{1}{1-\Gamma}}\).}
        \label{fig: novikov-polytrope-plots}
    \end{figure*}
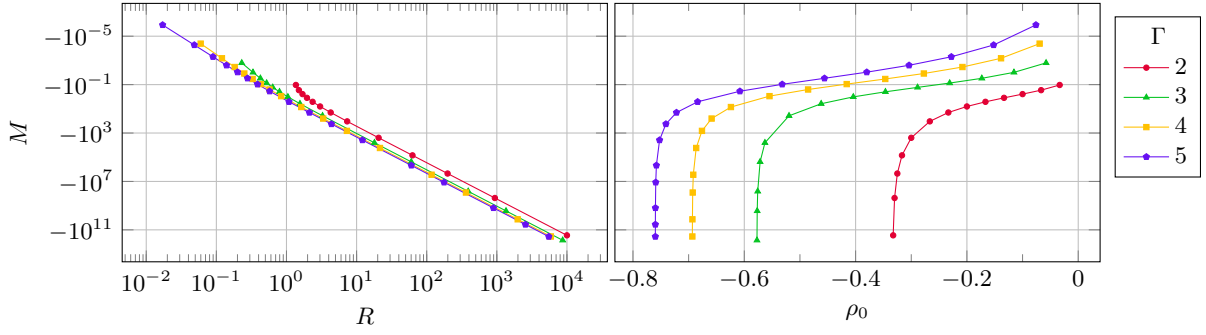

    It is interesting to point out that the models with the equation of state (\ref{eq: novikov-polytrope}) consistently violate the NEC.

\section{Stability of Negative-Mass Stars}\label{sec: stability-negative-mass}
The steps needed to perform the stability analysis of negative-mass stars are now the following:
\begin{enumerate}
    \item choose a density profile or an equation of state (in which case a value for the central density will also be needed);
    \item solve the TOV equation with this information and obtain the functions \(p_{\star}\), \(q_{\star}\), and \(w_{\star}\) of section \ref{sec: stability};
    \item if the model came from a density profile, verify that the parametric plot of \((\rho(r),P(r))\) is the graph of a function to ensure the model can be assumed to be barotropic;
    \item employ the variational or the Bardeen technique to determine if the star is stable.
\end{enumerate}

For some of the models we want to consider, the trial functions of Eq. (\ref{eq: trial-function}) will not be sufficient to establish instability. Hence, we define a new family of trial functions. To do so, we will first define the auxiliary function
\begin{equation}
    \theta_0(x) = \begin{cases}
        0, &\text{if } x \leq 0, \\
        \exp\pqty{-\frac{1}{x}}, &\text{if } x >0.
    \end{cases}
\end{equation}
\(\theta_0\) is a textbook example of a function that is everywhere smooth but fails to be analytic at \(x = 0\). Using \(\theta_0\) we can construct the functions
\begin{equation}\label{eq: bump-function}
    \theta_{a,b}(x) = \frac{\theta_0(a^2 - (x-b)^2)}{\theta_0(a^2)}.
\end{equation}
These are standard examples of smooth bump functions. These functions are everywhere smooth but have compact support. That is, \(\theta_{a,b}(x) = 0\) for any \(x\) such that \(\abs{x-b} > a\). For our purposes, we will consider bump functions with \(0 < b-a < b+a < R\), which ensures they are non-vanishing only in the interior of the star under consideration. Since the function and all of its derivatives vanish at the origin and the stellar surface, the boundary conditions for the pulsation equation are satisfied.

\subsection{Models with Density Profiles}
Variational stability tests for the models arising from density profiles are shown in table \ref{tab: stability}. It turns out that, for many models considered in table \ref{tab: profiles}, the trial function (\ref{eq: trial-function}) with \(n=2\) was sufficient to establish instability. Recall that all of these models violate the NEC consistently inside the star: they are such that \(P(r) + \rho(r) < 0\) for all \(r < R\) (but possibly not at \(r=R\)). Nevertheless, all of these models have \(\dv*{P}{\rho} < 0\), which renders \(p_{\star}\) negative. As a result, we cannot apply Bardeen's technique.

\begin{table}[t]
    \centering
    \caption{Values of \(n\) for the trial function of Eq. (\ref{eq: trial-function}) and of the quantity \(S[\psi_n]\) defined on Eq. (\ref{eq: S-psi}) for different profile choices \(\rho\). The leftmost column labels each stellar model following table \ref{tab: profiles} to facilitate discussion in the main text.}
    \label{tab: stability}
    \begin{ruledtabular}
    \begin{tabular}{cccc}
        model & \(\rho(r)\) & \(n\) & \(S[\psi_n]\) \\ \hline
        \ref*{model:r-1} & \(r-1\) & \num{2} & \num{-2.234e1} \\
        \ref*{model:r-2} & \(r-2\) & \num{2} & \num{-5.964e2} \\
        \ref*{model:r2-1} & \(r^2-1\) & \num{2} & \num{-1.095e2} \\
        \ref*{model:r2-2} & \(r^2-2\) & \num{2} & \num{-1.302e3} \\
        \ref*{model:r3-1} & \(r^3-1\) & \num{2} & \num{-3.813e2} \\
        \ref*{model:r3-2} & \(r^3-2\) & \num{2} & \num{-3.705e3} \\
        \ref*{model:exp-e} & \(\exp(r)-e\) & \num{2} & \num{-1.485e2} \\
        \ref*{model:exp-2e} & \(\exp(r)-2e\) & \num{2} & \num{-6.566e3} \\
        \ref*{model:-cos} & \(-\cos(\frac{\pi r}{2})\) & \num{2} & \num{-8.304e1}
    \end{tabular}
    \end{ruledtabular}
\end{table}

The trial function (\ref{eq: trial-function}) does not yield a good bound for models that partially satisfy the NEC (models \ref{model:-r} to \ref{model:12cos} in table \ref{tab: profiles}). For them, a better result is obtained using the bump function (\ref{eq: bump-function}). The results are shown in table \ref{tab: stability-bump}. All models are unstable. Since \(p_{\star}\) changes its sign inside the star, we cannot apply Bardeen's technique.

\begin{table}[t]
    \centering
    \caption{Values of \(a\) and \(b\) for the trial function of Eq. (\ref{eq: bump-function}) and of the quantity \(S[\theta_{a,b}]\) defined on Eq. (\ref{eq: S-psi}) for different profile choices \(\rho\). The leftmost column labels each stellar model following table \ref{tab: profiles} to facilitate discussion in the main text.}
    \label{tab: stability-bump}
    \begin{ruledtabular}
    \begin{tabular}{ccccc}
        model & \(\rho(r)\) & \(a\) & \(b\) & \(S[\theta_{a,b}]\) \\ \hline
        \ref*{model:-r} & \(-r\) & \num{1.00e-1} & \num{1.50e-1} & \num{-5.603e2} \\
        \ref*{model:-r2} & \(-r^2\) & \num{1.00e-1} & \num{3.00e-1} & \num{-5.412e2} \\
        \ref*{model:-r3} & \(-r^3\) & \num{1.00e-1} & \num{5.00e-1} & \num{-1.004e3} \\
        \ref*{model:-exp} & \(-\exp(r)\) & \num{1.00e-1} & \num{1.50e-1} & \num{-4.705e3} \\
        \ref*{model:-sin}& \(-\sin(\frac{\pi r}{2})\) & \num{1.00e-1} & \num{1.25e-1} & \num{-6.186e2} \\
        \ref*{model:12-r} & \(\frac{1}{2}-r\) & \num{1.25e-1} & \num{6.00e-1} & \num{-2.881e2} \\
        \ref*{model:25-r2} & \(\frac{2}{5}-r^2\) & \num{6.25e-2} & \num{7.50e-1} & \num{-5.095e3} \\
        \ref*{model:12cos} & \(\frac{1}{2}\cos(\pi r)\) & \num{1.00e-1} & \num{6.25e-1} & \num{-1.293e3} \\
    \end{tabular}
    \end{ruledtabular}
\end{table}

\subsection{Schwarzschild Star}\label{subsec: schwarzschild-star-stability}
Tables \ref{tab: stability} and \ref{tab: stability-bump} address almost all models of table \ref{tab: profiles}, the exception being model \ref{model:const}---the Schwarzschild star. A naive calculation for model \ref{model:const} would lead to infinite \(\dv*{P}{\rho}\) and to an ill-definition of \(p_{\star}\). This is because \(\dv*{\rho}{r}\) is zero since the density is constant. Nevertheless, one can treat \(\qty(\pdv*{P}{\rho})_s\) as finite by considering it corresponds to the value of the perturbations and that these perturbations do not have constant energy density. In Refs. \cite{chandrasekhar1964DynamicalInstabilityApJ,chandrasekhar1964DynamicalInstabilityPRL}, Chandrasekhar does exactly this by treating \(\gamma\) (as defined in Eq. (\ref{eq: gamma-definition})) as a constant. He then finds there is a critical value for \(\gamma\) which determines the onset of stability. We can establish a similar result. Define \(\bar{p}_{\star}\) through \(p_{\star}(r) = \gamma \bar{p}_{\star}(r)\). For the Schwarzschild star, one can compute \(\bar{p}_{\star}\), \(q_{\star}\), and \(w_{\star}\) exactly. The important properties about them are that \(\bar{p}_{\star}\) and \(q_{\star}\) are negative throughout the interior of the star, \(\bar{p}_{\star}\) only vanishes at the boundary, and \(q_{\star}\) never vanishes. The variational criterion establishes then that the star will be unstable if and only if,
\begin{equation}
    \frac{\int_0^R \bar{p}_{\star}(r) \psi'(r)^2 \dd{r} \gamma - \int_0^R q_{\star}(r)\psi(r)^2 \dd{r}}{\int_0^R w_{\star}(r) \psi(r)^2 \dd{r}} < 0
\end{equation}
for some \(\psi \neq 0\) satisfying the boundary conditions. This means the necessary and sufficient condition for stability is 
\begin{equation}
    \gamma > \frac{\int_0^R q_{\star}(r)\psi(r)^2 \dd{r}}{\int_0^R \bar{p}_{\star}(r) \psi'(r)^2 \dd{r}}
\end{equation}
for some \(\psi \neq 0\) satisfying the boundary conditions. The expression on the right-hand side can never vanish. Hence, we can define a critical value for \(\gamma\) by
\begin{equation}
    \gamma_c = \inf_{\psi} \frac{\int_0^R q_{\star}(r)\psi(r)^2 \dd{r}}{\int_0^R \bar{p}_{\star}(r) \psi'(r)^2 \dd{r}}
\end{equation}
where in the infimum it is understood that \(\psi \neq 0\) and that the boundary conditions are satisfied. 

By considering the trial functions (\ref{eq: trial-function}), one can show that 
\begin{equation}
    \gamma_c \leq \inf_n \frac{\int_0^R q_{\star}(r)\psi_n(r)^2 \dd{r}}{\int_0^R \bar{p}_{\star}(r) \psi_n'(r)^2 \dd{r}} \approx \num{0.777774}
\end{equation}
The actual infimum over \(n>0\) can also be found analytically and occurs in the limit \(n \to 0\).

The bump functions (\ref{eq: bump-function}) generate much better estimates. Fix \(\epsilon > 0\) and denote \(I_{\epsilon} = (\epsilon,R-\epsilon)\). Assume \(\epsilon < b-a < b+a < R - \epsilon\), which means \(\theta_{a,b}\) is compactly supported in \(I_{\epsilon}\). In this case, we have
\begin{subequations}
    \begin{align}
        \gamma_c &\leq \frac{\int_0^R q_{\star}(r)\theta_{a,b}(r)^2 \dd{r}}{\int_0^R \bar{p}_{\star}(r) \theta_{a,b}'(r)^2 \dd{r}}, \\
        &\leq \frac{\max_{I_{\epsilon}} q_{\star}(r) \int_0^R \theta_{a,b}(r)^2 \dd{r}}{\min_{I_{\epsilon}} \bar{p}_{\star}(r) \int_0^R \theta_{a,b}'(r)^2 \dd{r}}, \\
        &= c \frac{\int_0^R \theta_{a,b}(r)^2 \dd{r}}{\int_0^R \theta_{a,b}'(r)^2 \dd{r}},
    \end{align}
\end{subequations}
where \(c\) is the finite and positive constant obtained from the ratio between the minimum and maximum. The minimum and maximum are finite and non-vanishing because we assumed only a compact region of the integral contributes. Notice that the value of \(b\) does not change the integral as long as we integrate over the whole support of \(\theta_{a,b}\). This leads to 
\begin{equation}\label{eq: gamma-c-Ia}
    \gamma_c \leq c \frac{\int_{-a}^{+a} \theta_{a,0}(r)^2 \dd{r}}{\int_{-a}^{+a} \theta_{a,0}'(r)^2 \dd{r}} \leq c I(a),
\end{equation}
where \(I(a)\) is defined as
\begin{equation}\label{eq: Ia}
    I(a) = \frac{\int_{-a}^{+a} \theta_{a,0}(r)^2 \dd{r}}{\int_{-a}^{+a} \theta_{a,0}'(r)^2 \dd{r}}.
\end{equation}

These integrals can be evaluated analytically in terms of Meijer G-functions~\cite{andrews1985SpecialFunctionsEngineers,bateman1954HigherTranscendentalFunctionsVol1,luke1969SpecialFunctionsTheirVol1}. See App. \ref{app: bump-functions}. One finds that, for \(a \ll 1\), \(I(a)\) behaves as \(I(a) = a^4 + \order{a^6}\). Hence, by taking a sequence of bump functions with smaller values of \(a\) we can bound \(\gamma_c\) by increasingly smaller values. In the limit when we take the infimum, we find that \(\gamma_c = 0\).

Since \(\gamma_c\) vanishes, then the fact that the Schwarzschild star consistently violates the NEC together with Eq. (\ref{eq: gamma-adiabatic}) implies that stability requires
\begin{equation}
    \qty(\pdv{P}{\rho})_s > 0.
    \label{dpdrhos}
\end{equation}
Eq. (\ref{dpdrhos}) is a generic feature of stable stars, and we will comment on this later in this section.

\subsection{Models with an Equation of State}
Stability tests for the models arising from the equation of state (\ref{eq: novikov-polytrope}) are shown in table \ref{tab: stability-novikov-polytrope}. All tests indicated instability. For fixed \(\Gamma\) and \(n\), the value of \(S[\psi_n]\) appeared to increase in absolute value as \(\rho_0/\rho_c\) increased. This suggests that more negative mass leads to less instability since \(\sigma^2\) is bounded farther away from zero. Larger values of \(\sigma\) in absolute value are associated with faster exponential growth of the perturbations.

\begin{table}[t]
    \centering
    \caption{Values of \(n\) for the trial function of Eq. (\ref{eq: trial-function}) and of the quantity \(S[\psi_n]\) defined on Eq. (\ref{eq: S-psi}) for stars with equation of state of the form (\ref{eq: novikov-polytrope}) and different central densities \(\rho_0\). The central densities are given in units of the critical central density defined in Eq. (\ref{eq: novikov-polytrope-rho-crit}).}
    \label{tab: stability-novikov-polytrope}
    \begin{ruledtabular}
    \begin{tabular}{cccc}
        \(\Gamma\) & \(\rho_0/\rho_c\) & \(n\) & \(S[\psi_n]\) \\ \hline
        \num{2} & \num{0.1} & \num{2} & \num{-2.608e-1} \\ 
        \num{2} & \num{0.5} & \num{2} & \num{-6.057} \\ 
        \num{2} & \num{0.9} & \num{2} & \num{-1.040e3} \\ 
        \num{3} & \num{0.1} & \num{2} & \num{-1.075} \\ 
        \num{3} & \num{0.5} & \num{2} & \num{-9.669} \\ 
        \num{3} & \num{0.9} & \num{2} & \num{-3.910e2} \\
        \num{4} & \num{0.1} & \num{2} & \num{-2.189} \\ 
        \num{4} & \num{0.5} & \num{2} & \num{-1.412e1} \\ 
        \num{4} & \num{0.9} & \num{2} & \num{-2.827e2} \\
        \num{5} & \num{0.1} & \num{2} & \num{-3.426} \\ 
        \num{5} & \num{0.5} & \num{2} & \num{-1.927e1} \\ 
        \num{5} & \num{0.9} & \num{2} & \num{-2.427e2} 
    \end{tabular}
    \end{ruledtabular}
\end{table}

All of these models have negative values for \(p_{\star}\) due to the fact that \(\qty(\pdv*{P}{\rho})_s < 0\). Hence, we cannot apply Bardeen's method.

\subsection{Missing Examples}\label{subsec: missing-examples}
When dealing with the Schwarzschild star, we found that it is possible to have stability if \(\qty(\pdv*{P}{\rho})_s > 0\), but all the models we considered so far---both coming from profiles and from equations of state---violate this condition at some point in the star. Such violations can be understood under the assumption of the absence of singularities and a barotropic equation of state, in which case \(\qty(\pdv*{P}{\rho})_s > 0\) implies negative energies are only possible if there are also negative pressures.

The reason is as follows. First, notice that the condition~\eqref{dpdrhos} implies that pressure is a monotonically increasing function of the energy density. For a hypothetical star with zero density, we can readily see from the TOV equation that the pressure must vanish identically: 
\begin{equation}
    \dv{P}{r} = - 4 \pi P^2 r,
\end{equation}
where we used the absence of singularities to ensure that \(m(0) = 0\), and thus \(m(r) = 0\) everywhere. Since the equation of state is barotropic, we can write
\begin{equation}
    \dv{\rho}{r} = - 4 \pi P^2 r \qty(\pdv{\rho}{P})_s = 0,
\end{equation}
where the expression vanishes by the assumption that \(\rho = 0\) identically. If we assume \(\qty(\pdv*{\rho}{P})_s \neq 0\), we conclude \(P = 0\). The universality of the equation of state implies that the pressure must vanish when the energy density vanishes. Since \(P(\rho)\) is increasing, we conclude that \(+\infty > \qty(\pdv*{P}{\rho})_s > 0\) demands negative pressures for negative energy densities. This also means the NEC must be violated at all points with negative energy densities. In fact, \(\qty(\pdv*{P}{\rho})_s \geq 0\) is all one needs to conclude that the NEC is violated.

This can be summarized in the following proposition.

\begin{proposition}\label{prop: EoS-TOV}
    Consider a non-singular stationary spherically symmetric star in general relativity described by the Tolman--Oppenheimer--Volkoff equation. Suppose the star satisfies a barotropic and differentiable equation of state \(P = P(\rho)\) (and, consequently, \(\dv*{P}{\rho} = \qty(\pdv*{P}{\rho})_s\) is nowhere infinite). Suppose further that the equation of state admits a vacuum solution with \(\rho = 0\) identically on a neighborhood of the origin. Then the following statements hold true.
    \begin{enumerate}
        \item If the energy density vanishes in an interval \([0,a)\), \(0 < a < R\), then the pressure must vanish in the same interval.
        \item The equation of state must satisfy \(P(\rho=0) = 0\).
        \item If \(\qty(\pdv*{P}{\rho})_s \geq 0\), then the NEC is violated at points with negative energy densities.
    \end{enumerate}
\end{proposition}

Proposition \ref{prop: EoS-TOV} may seem in conflict with the results concerning stars obtained from a profile reported in section \ref{sec: examples}. Namely, figure \ref{fig: fit-eos} clearly shows an equation of state with \(P(\rho = 0) \neq 0\) which satisfies the TOV equation. However, as explained in section \ref{sec: examples}, stars with a sign flip originating from a choice of profile have extremely fine-tuned equations of state. Such equations of state may only hold for the particular solution to the TOV equations with a certain central density, in contradiction with the universality condition used in Proposition~\ref{prop: EoS-TOV}. They may admit no solutions in the physically meaningful case in which the central density is \(\rho(0) = 0\) (which should describe vacuum), and surely do not admit solutions in the case with \(\rho = 0\) identically. This behavior is also observed in an equation of state such as \(P = K_1\rho^4 + K_2\rho^2 + P_0\) for constant \(K_1 < 0\), \(K_2 > 0\), and \(P_0>0\), which does not admit a real solution for \(\rho(0) = 0\).

Can a star with \(\qty(\pdv*{P}{\rho})_s > 0\) have negative mass and negative pressure? Consider a star with negative pressure throughout the interior: \(P(r) < 0\) for all \(r \in [0,R)\). At the center, we have \(P(0) = P_0 < 0\) and \(\rho(0) = \rho_0 < 0\). Then the TOV equation tells us that, close to the center,
\begin{equation}
    \dv{P}{r} = - (P + \rho) \frac{4 \pi P r^3 + m}{r [r- 2m]} \leq 0,
\end{equation}
because \(P\), \(\rho\), and \(m\) are all negative there. We also have \(\qty(\pdv*{\rho}{P})_s > 0\), \(\dv*{\rho}{r} \leq 0\) as well. Hence, \(P\) and \(\rho\) can only become more negative as one gets farther from the center. Hence, there is no point \(R\) with \(P(R) = 0\) and the star is infinite. If we only assume \(\qty(\pdv*{P}{\rho})_s \geq 0\), we can adapt the argument by noticing that continuity impedes the positivity of \(\dv*{\rho}{r}\) at the points in which \(\qty(\pdv*{\rho}{P})_s\) diverges.

We can understand this result intuitively. In a negative-mass star, gravity is responsible for expanding the star, while pressure is responsible for contracting it. For pressure to contract the star, the pressure gradient must point inward. Hence, if pressure is negative, the star cannot be finite because the pressure would need to grow from a negative value up to zero close to the border, and hence the pressure gradient would point in the wrong direction close to the border.

The following proposition makes this result rigorous by means of the TOV equation.

\begin{proposition}\label{prop: NEC-boundary}
        Consider a static spherically symmetric star in general relativity with radius \(R > 0\) defined by the condition that \(P(R) = 0\), with \(P(r) < 0\) for all \(R - \epsilon < r < R\), for some \(\epsilon > 0\). Assume that the star's pressure \(P\) is differentiable in \([R-\epsilon,R]\) and that the energy density \(\rho\) is continuous in this interval. If the star has a negative total mass, the NEC must be satisfied in a subinterval of \([R-\epsilon,R]\).
    \end{proposition}
\begin{proof}
    From the conditions \(m(R) < 0\) and \(P(R) = 0\),
    \begin{equation}
        4\pi P(R) R^3 + m(R) < 0
    \end{equation}
    and 
    \begin{equation}
        R - 2m(R) > 0.
    \end{equation}

    Since \(P(r) < 0\) for all \(R - \epsilon < r < R\), we know that \(\dv*{P}{r} \geq 0\) at \(R\). Hence, the TOV equation reads
    \begin{equation}
        \eval{\dv{P}{r}}_{R} = - (P(R)+\rho(R)) \frac{4\pi P(R) R^3 + m(R)}{R [R - 2m(R)]} \geq 0.
    \end{equation}
    We see then that \(P(R) + \rho(R) \geq 0\).  

    If \(P(R) + \rho(R) > 0\), the result follows from continuity. Assume next that \(P(R) + \rho(R) = 0\). Since \(m(R) < 0\), continuity implies there is some \(\delta > 0\) (\(\delta < \epsilon\)) such that \(m(r) < 0\) for \(r \in (R-\delta,R]\). For \(r \in (R-\delta,R) \subsetneq (R-\epsilon,R)\) we know that \(P(r) < 0\). In particular,
    \begin{equation}
        P\qty(R - \frac{\delta}{2}) < 0.
    \end{equation}
    Since \(P(R) = 0\), the mean value theorem ensures that there is some point \(r^* \in \qty(R - \frac{\delta}{2},R)\) with \(P'(r^*) > 0\). At this point, we know that \(4\pi P(r^*) (r^*)^3 + m(r^*) < 0\). The TOV equation then enforces that \(P(r^*) + \rho(r^*) > 0\). Continuity ensures positivity holds in an interval.
\end{proof}

Hence, we learn that a negative mass star with negative pressure near the border must satisfy the NEC near the boundary. This contradicts our desire for a barotropic star with \(\qty(\pdv*{P}{\rho})_s \geq 0\) and \(P(\rho=0) = 0\). After all, there is at least one interval inside the star in which the pressure is negative, but the energy density is positive. We conclude that the barotropic scenarios with negative pressure and \(\qty(\pdv*{P}{\rho})_s \geq 0\) cannot lead to negative-mass stars. Since the star ends when the pressure changes sign, we conclude \(\qty(\pdv*{P}{\rho})_s < 0\) somewhere in the star.

These remarks can be summarized in the following proposition.

\begin{proposition}\label{prop: EoS-neg-mass}
    Consider a finite non-singular stationary spherically symmetric star in general relativity described by the Tolman--Oppenheimer Volkoff equation and satisfying the hypotheses of Proposition \ref{prop: NEC-boundary}. In particular, this implies the star is composed of a perfect fluid. Assume the fluid obeys a differentiable barotropic equation of state as in Proposition \ref{prop: EoS-TOV}. Then the star cannot have negative mass if the equation of state is such that \(\qty(\pdv*{P}{\rho})_s \geq 0\) everywhere inside the star.
\end{proposition}

It should be noted that if \(\qty(\pdv*{P}{\rho})_s < 0\) anywhere inside the star, then there is a region with \(p_{\star} < 0\), as one can tell from Eq. (\ref{eq: SLP-p-original}). This leads to two important observations.
\begin{enumerate}
    \item Bardeen's technique is generally not applicable to negative mass stars, as it relies on the assumption that \(p_{\star} > 0\) throughout the interior of the star (this is a condition for the Sturm Comparison Theorem).
    \item One will often be able to use a small bump function with large enough derivatives to make \(S[\theta_{a,b}]\) negative by restricting the support of \(\theta_{a,b}\) to the region in which \(p_{\star}\) is negative.
\end{enumerate}

This second observation corroborates that negative mass stars cannot be stable in general relativity and it is the basis for our main result, which is the following theorem.

\begin{theorem}\label{thm: main-theorem}
    Consider a finite non-singular stationary spherically symmetric star in general relativity described by the Tolman--Oppenheimer--Volkoff equation and satisfying the hypotheses of Proposition \ref{prop: NEC-boundary}. In particular, this implies the star is composed of a perfect fluid. Assume one of the two following conditions.
    \begin{enumerate}
        \item \(\qty(\pdv*{P}{\rho})_s < 0\) at some points of the star.
        \item The star obeys a differentiable barotropic equation of state as in Proposition \ref{prop: EoS-TOV} and it has negative total mass.
    \end{enumerate}
    Then the star is dynamically unstable. 
\end{theorem}
\begin{proof}
    If the second condition holds, Proposition \ref{prop: EoS-neg-mass} ensures that \(\qty(\pdv*{P}{\rho})_s < 0\) at some points inside the star. Hence, the first condition holds.
    
    Using continuity, we may assume that one of the points with \(\qty(\pdv*{P}{\rho})_s < 0\) is in the open interval \((0,R)\) without loss of generality. Denote this point by \(r^*\). Eq. (\ref{eq: SLP-p-original}) ensures \(p_{\star}(r^*) < 0\) \footnote{\(p_{\star}\) vanishes in the particular scenario in which \(\qty(\pdv*{P}{\rho})_s < 0\) only at points with \(P+\rho = 0\). Nevertheless, \(P+\rho = 0\) means the pressure is constant by the TOV equation. Furthermore, since \(\qty(\pdv*{P}{\rho})_s < 0\), applying the chain rule to the TOV equation also leads to the conclusion that \(\dv*{\rho}{r} = 0\). Hence, if \(\qty(\pdv*{P}{\rho})_s < 0\) at a point with \(P+\rho = 0\), then the star will be infinite, in contradiction with the hypotheses of the theorem.}. By continuity, we have the stronger result that there is an open interval \(I \subset (0,R)\) containing \(r^*\) in which \(p_{\star}\) is strictly negative. 

    The variational criterion establishes that the star will be unstable if there is any \(\zeta\) satisfying the boundary conditions~\eqref{eq: boundary-SLP} such that
    \begin{equation}\label{eq: instability-final-argument}
        \int_0^R p_{\star}(r) \qty(\zeta'(r))^2 \dd{r} - \int_0^R q_{\star}(r) \qty(\zeta(r))^2 \dd{r} < 0.
    \end{equation}
    This equation is a direct consequence of the expression \(S[\zeta] < 0\).
    
    Choose \(\zeta\) to be a bump function of the form (\ref{eq: bump-function}) with support in \(I\). Let us say \(\zeta = \theta_{a,b}\) for concreteness. Then Eq. (\ref{eq: instability-final-argument}) is equivalent to
    \begin{equation}
        1 > \frac{\int_c^d q_{\star}(r) \qty(\theta_{a,b}(r))^2 \dd{r}}{\int_c^d p_{\star}(r) (\theta_{a,b}'(r))^2 \dd{r}}.
    \end{equation}
    Notice, however, that 
    \begin{equation}
        \frac{\max q_{\star} \int_c^d \qty(\theta_{a,b}(r))^2 \dd{r}}{\min p_{\star} \int_c^d (\theta_{a,b}'(r))^2 \dd{r}} > \frac{\int_c^d q_{\star}(r) \qty(\theta_{a,b}(r))^2 \dd{r}}{\int_c^d p_{\star}(r) (\theta_{a,b}'(r))^2 \dd{r}},
    \end{equation}
    where the maximum and minimum are taken over \((c,d)\). Therefore, a sufficient condition for instability is that
    \begin{equation}
        1 > \frac{\max q_{\star} \int_c^d \qty(\theta_{a,b}(r))^2 \dd{r}}{\min p_{\star} \int_c^d (\theta_{a,b}'(r))^2 \dd{r}},
    \end{equation}
    where \(c < b-a < b+a < d\) by hypothesis. As in the Schwarzschild star case, we can now write this expression as
    \begin{equation}
        1 > \frac{\max q_{\star} \int_{-a}^a \qty(\theta_{a,0}(r))^2 \dd{r}}{\min p_{\star} \int_{-a}^a (\theta_{a,0}'(r))^2 \dd{r}},
    \end{equation}
    because we are still integrating over the whole support of \(\theta_{a,b}\) and the value of \(b\) does not modify the value of the integrals. 

    Recalling Eq. (\ref{eq: Ia}), we see thus that the very general star we are considering will be unstable if there is at least one value of \(0 < a < R/2\) for which  
    \begin{equation}\label{eq: criterion-main-theorem}
        1 > \frac{\max q_{\star}}{\min p_{\star}} I(a).
    \end{equation}
    Nevertheless, as mentioned in Section \ref{subsec: schwarzschild-star-stability}, \(I(a) = a^4 + \order{a^6}\) for small values of \(a\), which establishes \(I(a)\) can be made as small as desired by picking sufficiently small values of \(a\). Hence, it is always possible to choose \(a\) such that Eq. (\ref{eq: criterion-main-theorem}) holds, which establishes that the class of stars considered in this theorem is always unstable. 
\end{proof}

In practice, what the proof teaches us is that the regions of the star with \(\qty(\pdv*{P}{\rho})_s < 0\) are unstable. By carefully exciting these regions with a perturbation, we can disrupt the equilibrium of the star. This does not require any assumption on the sign of the energy density. 

One may notice Proposition \ref{prop: EoS-TOV} (and therefore Theorem \ref{thm: main-theorem}) assumes ``barotropic'' to mean a differentiable equation of state of the form \(P = P(\rho)\). This rules out the limiting case \(\rho = \rho(P) = \rho_0\) for a constant \(\rho_0\). This is the case, for example, for the Schwarzschild star. As discussed in section \ref{subsec: schwarzschild-star-stability}, there are situations in which the pulsation equation suggests the Schwarzschild star could be stable. Nevertheless, since the key quantity in the stability analysis, \(\qty(\pdv*{P}{\rho})_s\), is ill defined in this case, the Chandrasekhar pulsation equation should be viewed with skepticism. In scenarios where \(\qty(\pdv*{\rho}{P})_s=0\), it may be more meaningful to rederive the pulsation equation with the assumption that the perturbation in the energy density always vanishes. Indeed, if the equation of state demands \(\rho(P) = \rho_0\) in a region, then the perturbation cannot change the value of \(\rho_0\) in said region. Since this limiting case is not of particular physical interest, we will not consider it in more depth.

\subsection{Instability Timescales}
    Finally, we estimate the timescales associated with the hydrodynamical instabilities. For this, we must recover the units.

    Let \(N(\sigma)\) denote the absolute value of the numerical value obtained for \(\sigma\) in the dimensionless system of units we have been using so far. To restore dimensions, we will use appropriate powers of \(c\), \(G\), and some length scale \(R_0\) characterizing the system. The value of \(\abs{\sigma}\) in an arbitrary system of units will then be given by
    \begin{equation}
        \abs{\sigma} = N(\sigma) \frac{c}{R_0}.
    \end{equation}
    
    Accordingly, the instability happens in the timescale
    \begin{equation}
        T = \frac{R_0}{N(\sigma) c},
    \end{equation}
    which is the time it takes for light to cross a length \(R_0\) corrected by the numerical factor associated with the instability. 

    According to tables \ref{tab: stability}, \ref{tab: stability-bump}, and \ref{tab: stability-novikov-polytrope}, the smallest value we found for \(\sigma^2\) was of order unity (\(\Gamma = 3\), \(\rho_0 = 0.1\rho_c\) in table \ref{tab: stability-novikov-polytrope}). Hence, the smallest value we found for \(N(\sigma)\) was at order \(N(\sigma) \sim 1\). Larger values would lead to smaller instability timescales, so this is the ``most stable scenario''. 

    We thus find that the instability timescales for the models we considered are such that
    \begin{equation}
        T \lesssim \frac{R_0}{c}.
    \end{equation}
    It would be incorrect to conclude from this expression that the instability is faster than light. One should recall that the instabilities we are considering are local in nature. They are derived from the pulsation equation, which describes how each fluid element moves around its equilibrium position. A way to picture this is to remember that Sturm--Liouville problems traditionally have an infinite and unbounded sequence of eigenvalues. Hence, even in a stable star, there are oscillation modes with arbitrarily large frequencies, which are thus associated with arbitrarily small oscillation periods. The phenomenon we are seeing here is the same. \(T\) does not measure the time it takes for the star to undergo gravitational collapse, for example, but rather the time it takes for the perturbations to grow significantly.

    For \(R_0 \sim \SI{7e8}{\meter}\) (about the size of the solar radius), \(T \sim \SI{2e2}{\second}\), indicating that only larger structures with negative masses can live longer than a couple of minutes, posing the astrophysical question of their formation.

\section{Conclusions}\label{sec: conclusions}
We have discussed the properties of negative-mass relativistic stars. In particular, we reviewed how quantum theory, grounded on Borde's and Penrose--Sorkin--Woolgar's theorems, hampers the existence of negative-mass stars. Furthermore, we found that classical general relativity abhors negative masses in the sense that all models of negative mass stars considered in this work---arising both from an equation of state or an energy density profile---turned out to be unstable \footnote{As mentioned at the end of Section \ref{subsec: missing-examples}, the Schwarzschild star and stars with \(\qty(\pdv*{\rho}{P})_s = 0\) may need to be treated separately. Nevertheless, this case may require modifying the Chandrasekhar pulsation equation and hence we chose not to consider it in depth, as it seems uninteresting.}. \emph{A fortiori}, we showed that any star with \(\qty(\pdv*{P}{\rho})_s < 0\) at some point must be dynamically unstable, and also that all barotropic negative-mass stars satisfy this condition. None of our stability analyses depends on the validity of energy conditions.

It is interesting to note the implications of these results for the cosmic-weight watcher conjecture about negative masses. Although general relativity cannot forbid an equilibrium solution with negative mass---after all, any Lorentzian geometry is a solution to the Einstein field equations for the appropriate stress tensor---it may forbid stable solutions with negative mass. This happens in the cases considered in this work. While one can construct a negative-mass star, one cannot expect it to be stable under small hydrodynamical perturbations. 

Notice this gives a classical answer to a problem that may appear to be quantum in nature. One would typically expect that the energy conditions of quantum theory would be necessary, not only sufficient, to forbid negative-mass relativistic stars. The argument would be essentially that quantum field theory gives input on the sorts of matter available to general relativity and thus forbids otherwise valid solutions. However, our results indicate that this is not the only mechanism used by the cosmic weight-watcher of Costa and Matsas to forbid negative masses. Stability is also an important mechanism that seems to rule out negative masses even at a classical level. 

\acknowledgments
We thank Daniel A. T. Vanzella, and George E. A. Matsas for illuminating discussions that initiated and fueled this project. Part of the calculations in this work were carried out with the aid of \texttt{Mathematica} \cite{wolframresearch2024Mathematica140} (and in particular the \texttt{OGRe} package \cite{shoshany2021OGReObjectOrientedGeneral}) in a license to the University of São Paulo (N. A. A.'s \emph{alma mater}). N. A. A. was supported by the Coordenação de Aperfeiçoamento de Pessoal de Nível Superior---Brasil (CAPES)---Finance Code 001.

\begin{appendix}
\section{ANEC in a Stellar Spacetime}\label{app: anec-stellar}
    In this appendix, we rewrite the ANEC integral in a way that is easier to compute in a spacetime representing a star. To do so, we begin by noticing that Eqs. (\ref{eq: fluid-stress-tensor}) and (\ref{eq: anec-integral}) imply that 
    \begin{equation}\label{eq: anec-integral-u-gamma}
        \int \tensor{T}{_a_b} \tensor{\dot{\gamma}}{^a}\tensor{\dot{\gamma}}{^b} \dd{\lambda} = \int (P+\rho) \qty(\tensor{u}{_a} \tensor{\dot{\gamma}}{^a})^2 \dd{\lambda}.
    \end{equation}
    
    Notice that
    \begin{equation}
        \tensor{\dot{\gamma}}{^a}(\lambda) = \dv{t}{\lambda}\tensor{\qty(\pdv{t})}{^a} + \dv{r}{\lambda}\tensor{\qty(\pdv{r})}{^a} + \dv{\varphi}{\lambda}\tensor{\qty(\pdv{\varphi})}{^a},
    \end{equation}
    where we are considering an arbitrary null geodesic and will impose the radial condition only later. Notice that spherical symmetry allowed us to make the simplifying assumption that \(\theta = \frac{\pi}{2}\) without any loss of generality. 
    
    The spacetime has a timelike Killing vector field and an axial Killing vector field (among others). This allows us to identify the conserved quantities 
    \begin{gather}
        \varepsilon = - \tensor{\qty(\pdv{t})}{^a}\tensor{\dot{\gamma}}{^b}\tensor{g}{_a_b} = e^{2\phi(r)} \dv{t}{\lambda} \\
        \intertext{and}
        \ell = \tensor{\qty(\pdv{\varphi})}{^a}\tensor{\dot{\gamma}}{^b}\tensor{g}{_a_b} = r^2 \dv{\varphi}{\lambda},
    \end{gather}
    which are interpreted as energy and angular momentum.
    
    It follows then that
    \begin{equation}
        \tensor{\dot{\gamma}}{^a}(\lambda) = \varepsilon e^{-2\phi(r)} \tensor{\qty(\pdv{t})}{^a} + \dv{r}{\lambda}\tensor{\qty(\pdv{r})}{^a} + \frac{\ell}{r^2}\tensor{\qty(\pdv{\varphi})}{^a},
    \end{equation}
    and hence, using Eq. (\ref{eq: four-velocity-tov}),
    \begin{equation}
        \tensor{u}{_a}\tensor{\dot{\gamma}}{^a} = - e^{-\phi(r)}\varepsilon.
    \end{equation}
    
    The condition \(\tensor{\dot{\gamma}}{^a}\tensor{\dot{\gamma}}{_a} = 0\) implies that 
    \begin{equation}\label{eq: drdlambda-anec-integral}
        \qty(\dv{r}{\lambda})^2 = \qty(1 - \frac{2m(r)}{r})\qty(\varepsilon e^{-2\phi(r)} - \frac{\ell^2}{r^2}).
    \end{equation}
    Notice that (assuming \(2m(r) < r\) for all \(r\), which holds in our cases of interest) the point of closest approach, \(r_{\text{min}}\), is characterized by
    \begin{equation}
        r_{\text{min}}^2 e^{-2\phi(r_{\text{min}})} = \frac{\ell^2}{\varepsilon^2},
    \end{equation}
    and if this equation admits multiple solutions, then \(r_{\text{min}}\) is the largest solution with \(r_{\text{min}} < R\).
    
    Let us focus on the time interval in which the coordinate \(r\) grows along the geodesic. We choose \(\lambda\) such that \(r(0) = r_{\text{min}}\), so that this corresponds to computing the integral
    \begin{equation}
        \int_0^{+\infty} \tensor{T}{_a_b} \tensor{\dot{\gamma}}{^a}\tensor{\dot{\gamma}}{^b} \dd{\lambda} = \int_0^{+\infty} (P+\rho) \qty(\tensor{u}{_a} \tensor{\dot{\gamma}}{^a})^2 \dd{\lambda},
    \end{equation}
    which we already know will yield 
    \begin{equation}
        \int_0^{+\infty} \tensor{T}{_a_b} \tensor{\dot{\gamma}}{^a}\tensor{\dot{\gamma}}{^b} \dd{\lambda} = \varepsilon^2 \int_0^{+\infty} (P+\rho) e^{-2\phi} \dd{\lambda}.
    \end{equation}
    We can change the integration parameter from \(\lambda\) to \(r\) by using Eq. (\ref{eq: drdlambda-anec-integral}) with the additional assumption that \(\dv*{r}{\lambda} > 0\). This yields, upon simplification,
    \begin{multline}
        \int_0^{+\infty} \tensor{T}{_a_b} \tensor{\dot{\gamma}}{^a}\tensor{\dot{\gamma}}{^b} \dd{\lambda} \\ = \varepsilon \int_{r_{\text{min}}}^{R} \frac{(P(r)+\rho(r)) e^{-\phi(r)} \dd{r}}{\sqrt{\qty(1 - \frac{2 m(r)}{r})\qty(1 - \frac{\ell^2}{\varepsilon^2}\frac{e^{2\phi(r)}}{r^2})}}.
    \end{multline}
    
    To get the full ANEC integral, we notice that the expression is symmetric on whether \(r\) is increasing or decreasing. Hence, we find that
    \begin{multline}\label{eq: ANEC-integral-in-TOV}
        \int_{-\infty}^{+\infty} \tensor{T}{_a_b} \tensor{\dot{\gamma}}{^a}\tensor{\dot{\gamma}}{^b} \dd{\lambda} \\ = 2 \varepsilon \int_{r_{\text{min}}}^{R} \frac{(P(r)+\rho(r)) e^{-\phi(r)} \dd{r}}{\sqrt{\qty(1 - \frac{2 m(r)}{r})\qty(1 - \frac{\ell^2}{\varepsilon^2}\frac{e^{2\phi(r)}}{r^2})}}.
    \end{multline}
    
    In the particular case with \(\ell = 0\), corresponding to a radial geodesic, we get \(r_{\min} = 0\) and 
    \begin{equation}\label{eq: radial-ANEC-integral-in-TOV}
        \int_{-\infty}^{+\infty} \tensor{T}{_a_b} \tensor{\dot{\gamma}}{^a}\tensor{\dot{\gamma}}{^b} \dd{\lambda} = 2 \varepsilon \int_0^{R} (P+\rho) \frac{e^{-\phi(r)}}{\sqrt{1 - \frac{2m(r)}{r}}} \dd{r}.
    \end{equation}

\section{NEC Violations in QFT}\label{app: nec-violations-qft}
    In this appendix included only in the \texttt{arXiv} version of the paper, we give a short, but standard, argument showing that one can violate the null energy condition in quantum field theory. The argument is similar to the one given in Refs. \cite{fewster2012LecturesQuantumEnergy,fewster2017QuantumEnergyInequalities}. In summary, we will find a state for which the expectation value of \(\normord{\tensor{\hat{T}}{_a_b}\tensor{k}{^a}\tensor{k}{^b}}\)---where \(\tensor{\hat{T}}{_a_b}\) is the stress tensor operator, \(\tensor{k}{^a}\) is a null vector, and the colons denote normal ordering---becomes negative. This appendix uses Planck units \(G = c = \hbar = 1\).

    We consider a minimally coupled massive scalar field in Minkowski spacetime. The stress-energy tensor for such a field is known to be 
    \begin{equation}
        \tensor{\hat{T}}{_a_b} = \tensor{\nabla}{_a}\hat{\varphi}\tensor{\nabla}{_b}\hat{\varphi} - \frac{1}{2}\tensor{\eta}{_a_b}[\tensor{\nabla}{_c}\hat{\varphi}\tensor{\nabla}{^c}\hat{\varphi} + m^2 \hat{\varphi}^2],
    \end{equation}
    where the hats remind us that we should consider these objects as operators. We will soon need to renormalize the stress tensor, which can be implemented by the normal-ordering prescription.

    For any null vector field \(\tensor{k}{^a}\),
    \begin{equation}\label{eq: Tab-ka-kb}
        \tensor{\hat{T}}{_a_b}\tensor{k}{^a}\tensor{k}{^b} = \tensor{k}{^a}\tensor{\nabla}{_a}\hat{\varphi}\tensor{k}{^b}\tensor{\nabla}{_b}\hat{\varphi} \equiv \hat{\mathcal{H}},
    \end{equation}
    where we defined the ``null Hamiltonian density'' \(\hat{\mathcal{H}}\) with the sole intention of simplifying the notation. As a side note, if we were in the classical theory Eq. (\ref{eq: Tab-ka-kb}) would be manifestly non-negative since we would have
    \begin{equation}
        \tensor{T}{_a_b}\tensor{k}{^a}\tensor{k}{^b} = (\tensor{k}{^a}\tensor{\nabla}{_a}\varphi)^2 \geq 0,
    \end{equation}
    showing that a minimally coupled classical scalar field always obeys the NEC.

    Our next step is to express the normal-ordered operator \(\normord{\hat{\mathcal{H}}}\) in terms of creation and annihilation operators. For this, we decompose the quantum field \(\hat{\varphi}\) in creation and annihilation operators by writing
    \begin{equation}\label{eq: varphi-ladders}
        \hat{\varphi}(x) = \frac{1}{\qty(2\pi)^{\frac{d}{2}}} \int \qty[\hat{a}_{\vb{p}} e^{i p \cdot x} + \hat{a}_{\vb{p}}^\dagger e^{-i p \cdot x}] \frac{\dd[d]{p}}{\sqrt{2 \omega_{\vb{p}}}},
    \end{equation}
    where \(p \cdot x = \tensor{p}{_\mu}\tensor{x}{^\mu}\),
    \begin{equation}
        \tensor{p}{^a}\tensor{p}{_a} = - (\tensor{p}{^0})^2 + \norm{\vb{p}}^2 = -m^2
    \end{equation}
    is understood, and we denote \(\omega_{\vb{p}} = \sqrt{\norm{\vb{p}}^2 + m^2}\). For generality, we are assuming at this stage a \(d+1\) dimensional spacetime.

    The decomposition in Eq. (\ref{eq: varphi-ladders}) leads to the canonical commutation relations
    \begin{equation}
        \comm{\hat{a}_{\vb{p}}}{\hat{a}_{\vb{q}}^\dagger} = \delta^{(d)}(\vb{p}-\vb{q}).
    \end{equation}

    From Eq. (\ref{eq: varphi-ladders}) we extract
    \begin{equation}\label{eq: k-nabla-varphi}
        \tensor{k}{^b}\tensor{\nabla}{_b}\hat{\varphi} = \frac{i}{\qty(2\pi)^{\frac{d}{2}}} \int \qty[\hat{a}_{\vb{p}} e^{i p \cdot x} - \hat{a}_{\vb{p}}^\dagger e^{-i p \cdot x}] \frac{\tensor{k}{^b}\tensor{p}{_b} \dd[d]{p}}{\sqrt{2 \omega_{\vb{p}}}}.
    \end{equation}

    From Eqs. (\ref{eq: Tab-ka-kb}) and (\ref{eq: k-nabla-varphi}) we find that 
    \begin{widetext}
    \begin{equation}
        \normord{\hat{\mathcal{H}}} = -\frac{1}{\qty(2\pi)^{d}} \int \qty[\hat{a}_{\vb{p}}\hat{a}_{\vb{q}} e^{i (p + q) \cdot x} - \hat{a}_{\vb{p}}^\dagger \hat{a}_{\vb{q}} e^{-i (p - q) \cdot x} - \hat{a}_{\vb{q}}^\dagger\hat{a}_{\vb{p}} e^{-i (q-p) \cdot x} + \hat{a}_{\vb{p}}^\dagger\hat{a}_{\vb{q}}^\dagger e^{-i (p + q) \cdot x}] \frac{(\tensor{k}{^b}\tensor{p}{_b})(\tensor{k}{^c}\tensor{q}{_c}) \dd[d]{p} \dd[d]{q}}{2 \sqrt{\omega_{\vb{p}} \omega_{\vb{q}}}}.
    \end{equation}
    \end{widetext}

    The next step is to find a state \(\ket{\psi}\) which yields \(\ev{\normord{\hat{\mathcal{H}}(x)}}{\psi} < 0\) in some region. We take
    \begin{equation}
        \ket{\psi} = \cos\beta \ket{0} + \sin\beta \ket{fg},
    \end{equation}
    where \(\beta\) is a real parameter which we will fix later, \(\ket{0}\) is the (Minkowski) vacuum, and 
    \begin{equation}
        \ket{fg} = \left(\int \frac{\sqrt{2 \omega_{\vb{p}}} f(\vb{p})}{\tensor{k}{^b}\tensor{p}{_b}} \hat{a}^\dagger_{\vb{p}} \dd[d]{p} \int \frac{\sqrt{2 \omega_{\vb{q}}} g(\vb{q})}{\tensor{k}{^c}\tensor{q}{_c}} \hat{a}^\dagger_{\vb{q}} \dd[d]{q}\right) \ket{0}
    \end{equation}
    is a smeared two-particle state. The smearing is necessary to later keep the expectation value of \(\normord{\hat{\mathcal{H}}}\) well-defined. The factors of \(\omega\) and \(\tensor{k}{^a}\tensor{p}{_a}\) are chosen for later convenience. The normalization of this state is given by
    \begin{equation}
        \bra{f_2 g_2}\ket{f_1 g_1} = (f_2,f_1) (g_2,g_1) + (f_2,g_1) (g_2,f_1),
    \end{equation}
    where 
    \begin{equation}
        (f, g) = \int f^*(\vb{p}) g(\vb{p}) \frac{2 \omega_{\vb{p}} \dd[d]{p}}{(\tensor{k}{^b}\tensor{p}{_b})^2}.
    \end{equation}
    For simplicity, we will take \(f = g\). Imposing \(\braket{ff} = 1\) implies
    \begin{equation}
        (f,f) = \frac{1}{\sqrt{2}}.
    \end{equation}
    
    We recall that the Minkowski vacuum is annihilated by all annihilation operators \(\hat{a}_{\vb{p}}\), and this implies
    \begin{equation}
        \ev{\normord{\hat{\mathcal{H}}(x)}}{0} = 0.
    \end{equation}
    We also have the normalization condition \(\braket{0} = 1\).

    We notice then that
    \begin{multline}
        \ev{\normord{\hat{\mathcal{H}}}}{\psi} = \sin(2\beta) \Re[\mel{0}{\normord{\hat{\mathcal{H}}}}{ff}] \\ + \sin^2\beta \ev{\normord{\hat{\mathcal{H}}}}{ff}.
    \end{multline}

    The matrix elements are given by
    \begin{equation}
        \mel{0}{\normord{\hat{\mathcal{H}}}}{ff} = - \frac{2}{(2\pi)^d} \int f(\vb{p}) f(\vb{q}) e^{i (p+q) \cdot x} \dd[d]{p} \dd[d]{q}.
    \end{equation}
    and
    \begin{equation}
        \ev{\normord{\hat{\mathcal{H}}}}{ff} = \frac{4}{\sqrt{2} (2\pi)^d} \int f^*(\vb{p})f(\vb{q}) e^{- i (p-q)\cdot x} \dd[d]{p} \dd[d]{q}.
    \end{equation}

    In total, we find that
    \begin{widetext}
        \begin{equation}\label{eq: ev-null-ham}
            \ev{\normord{\hat{\mathcal{H}}(x)}}{\psi} = - \frac{2 \sin(2 \beta)}{(2\pi)^d} \Re\qty[\int f(\vb{p}) f(\vb{q}) e^{i (p+q) \cdot x} \dd[d]{p} \dd[d]{q}] + \frac{4 \sin^2\beta}{\sqrt{2} (2\pi)^d} \int f^*(\vb{p})f(\vb{q}) e^{- i (p-q)\cdot x} \dd[d]{p} \dd[d]{q}.
        \end{equation}
    Once \(x\) is fixed, one can then choose \(\beta\) arbitrarily small so that the first term is negative (\(f\) can be chosen so that it does not vanish). In the small \(\beta\) limit, the second term becomes negligible, exhibiting local violations of the NEC. 

    Let us consider now what happens when we integrate along a null geodesic parallel to \(\tensor{k}{^a}\), which will allow us to consider the ANEC. For simplicity, we focus on a massive field (thus \(m>0\)) in a two-dimensional Minkowski spacetime (\(d = 1\)). A curve parallel to \(\tensor{k}{^a}\) is described in coordinates by \(\tensor{x}{^\mu} = \lambda \tensor{k}{^\mu}\), where \(\lambda\) is an affine parameter. 

    With this in mind, we see that Eq. (\ref{eq: ev-null-ham}) leads to the expression 
        \begin{multline}
            \int \ev{\normord{\hat{\mathcal{H}}}}{\psi} \dd{\lambda} = - \frac{2 \sin(2 \beta)}{(2\pi)^{d-1}} \Re\qty[\int f(\vb{p}) f(\vb{q}) \delta((p+q) \cdot k) \dd[d]{p} \dd[d]{q}] \\ + \frac{4 \sin^2\beta}{\sqrt{2} (2\pi)^{d-1}} \int f^*(\vb{p})f(\vb{q}) \delta((p-q) \cdot k) \dd[d]{p} \dd[d]{q}.
        \end{multline}
    \end{widetext}
    The first integral vanishes because the delta's argument,
    \begin{equation}
        (p+q)\cdot k = (\vb{p} + \vb{q}) \vdot \vb{k} - (\omega_{\vb{p}} + \omega_{\vb{q}}) \tensor{k}{^0},
    \end{equation}
    is strictly positive since \(\tensor{k}{^0} = \norm{\vb{k}}\) and \(\omega_{\vb{p}} > \norm{\vb{p}}\). Notice the main reason this quantity cannot vanish is the fact that \(\omega_{\vb{p}} > 0\), so we cannot have \(\tensor{p}{^a} = -\tensor{q}{^a}\) because \(\omega_{\vb{p}} = - \omega_{\vb{q}}\) is impossible.

    To deal with the second term, we use the simplifying assumption that \(d=1\). Then the ANEC integral becomes
    \begin{equation}
        \int \ev{\normord{\hat{\mathcal{H}}}}{\psi} \dd{\lambda} = \frac{4 \sin^2\beta \omega_p}{\sqrt{2}} \int \frac{f^*(p)f(p)}{\abs{\omega_p - p}} \dd{p},
    \end{equation}
    which is manifestly positive. Therefore, we see that the ANEC holds for the state \(\ket{\psi}\) of a massive scalar field in a two-dimensional Minkowski spacetime.

\section{Bump Function Integrals}\label{app: bump-functions}
    While dealing with the stability of the Schwarzschild star, we faced the integrals
    \begin{equation}
        \int_{-a}^{a} \theta_{a,0}(r)^2 \dd{r} \qq{and} \int_{-a}^{a} \theta'_{a,0}(r)^2 \dd{r}.
    \end{equation}
    In this appendix, which is included only in the \texttt{arXiv} version of the paper, we evaluate them exactly. We begin by defining
    \begin{equation}
        I_1(a) = \int_{-a}^{a} \theta_{a,0}(r)^2 \dd{r}
    \end{equation}
    and
    \begin{equation}
        I_2(a) = \int_{-a}^{a} \theta'_{a,0}(r)^2 \dd{r}.
    \end{equation}
    
    From the definition of \(\theta_{a,0}\) (Eq. \ref{eq: bump-function}), it is straightforward to show that
    \begin{equation}
        I_1(a) = e^{\frac{2}{a^2}} \int_{-a}^{a} e^{\frac{2}{r^2-a^2}} \dd{r}
    \end{equation}
    and
    \begin{equation}
        I_2(a) = 4 e^{\frac{2}{a^2}} \int_{-a}^{a} \frac{r^2 e^{\frac{2}{r^2-a^2}}}{(r^2-a^2)^4} \dd{r}.
    \end{equation}
    
    Using the fact that the integrands are even and performing the substitution \(r = ax\) we get that
    \begin{equation}
        I_1(a) = 2 a e^{\frac{2}{a^2}} \int_{0}^{1} e^{\frac{2}{a^2(x^2-1)}} \dd{x}
    \end{equation}
    and
    \begin{equation}
        I_2(a) = 8 a^{-5} e^{\frac{2}{a^2}} \int_{0}^{1} \frac{x^2 e^{\frac{2}{a^2(x^2-1)}}}{(x^2-1)^4} \dd{x}.
    \end{equation}
    
    Next, we make the substitutions 
    \begin{equation}
        u = - \frac{1}{x^2-1},
    \end{equation}
    which involve
    \begin{equation}
        \dd{x} = \frac{\dd{u}}{2 u \sqrt{u^2 - u}}.
    \end{equation}
    This leads to 
    \begin{equation}
        I_1(a) = a e^{\frac{2}{a^2}} \int_1^{+\infty} \frac{e^{\frac{-2u}{a^2}}}{u \sqrt{u^2 - u}} \dd{u}
    \end{equation}
    and
    \begin{equation}
        I_2(a) = 4 a^{-5} e^{\frac{2}{a^2}} \int_{1}^{+\infty} e^{\frac{-2u}{a^2}} u\sqrt{u^2 - u} \dd{u}.
    \end{equation}
    
    The problem of finding both \(I_1(a)\) and \(I_2(a)\) has now been reduced to computing integrals of the form
    \begin{equation}
        J(z;\alpha,\beta) \equiv \int_1^{+\infty} e^{- z u} u^{-\alpha} (u-1)^{\beta-1} \dd{u}.
    \end{equation}
    For \(I_1\) we have \(z = 2/a^2\), \(\alpha = 3/2\), and \(\beta = 1/2\). For \(I_2\), \(z = 2/a^2\), \(\alpha=-3/2\), and \(\beta = 3/2\). The trick to solve this integrals is the same one often used in symbolic integration: we express the integrals in terms of Meijer G-functions \cite{andrews1985SpecialFunctionsEngineers,bateman1954HigherTranscendentalFunctionsVol1,luke1969SpecialFunctionsTheirVol1} and then apply an integration theorem. 
    
    We thus begin by noticing that \cite{andrews1985SpecialFunctionsEngineers}
    \begin{equation}
        e^{-z u} = G^{1,0}_{0,1}\left(z u\middle| \genfrac{}{}{0pt}{}{}{0} \right),
    \end{equation}
    which leads us to 
    \begin{equation}
        J(z;\alpha,\beta) = \int_1^{+\infty} G^{1,0}_{0,1}\left(z u\middle| \genfrac{}{}{0pt}{}{}{0} \right) u^{\beta} (u-1)^{\gamma} \dd{u}.
    \end{equation}
    The result of this integral is known and given by \cite{bateman1954TablesIntegralTransformsVol2,luke1969SpecialFunctionsTheirVol1}
    \begin{equation}
        J(z;\alpha,\beta) = \Gamma(\beta) G^{2,0}_{1,2}\left(z \middle| \genfrac{}{}{0pt}{}{\alpha}{\alpha-\beta,0} \right),
    \end{equation}
    where \(\Gamma\) is the gamma function. Therefore, we conclude that
    \begin{equation}
        \int_1^{+\infty} \frac{e^{\frac{-2u}{a^2}}}{u \sqrt{u^2 - u}} \dd{u} = \sqrt{\pi} G^{2,0}_{1,2}\left(\frac{2}{a^2}\middle| \genfrac{}{}{0pt}{}{\frac{3}{2}}{0,1} \right)
    \end{equation}
    and
    \begin{equation}
        \int_1^{+\infty} e^{\frac{-2u}{a^2}} u \sqrt{u^2 - u} \dd{u} = \frac{\sqrt{\pi}}{2} G^{2,0}_{1,2}\left(\frac{2}{a^2}\middle| \genfrac{}{}{0pt}{}{-\frac{3}{2}}{-3,0} \right).
    \end{equation}
    
    Bringing everything together, we find that
    \begin{equation}
        I_1(a) = a \sqrt{\pi} e^{\frac{2}{a^2}} G^{2,0}_{1,2}\left(\frac{2}{a^2}\middle| \genfrac{}{}{0pt}{}{\frac{3}{2}}{0,1} \right)
    \end{equation}
    and
    \begin{equation}
        I_2(a) = 2 \sqrt{\pi} a^{-5} e^{\frac{2}{a^2}} G^{2,0}_{1,2}\left(\frac{2}{a^2}\middle| \genfrac{}{}{0pt}{}{-\frac{3}{2}}{-3,0} \right).
    \end{equation}
    
    Ultimately, the quantity we are interested in is the function
    \begin{subequations}
        \begin{align}
            I(a) &= \frac{I_1(a)}{I_2(a)}, \\
            &= \frac{a^6}{2} \frac{G^{2,0}_{1,2}\left(\frac{2}{a^2}\middle| \genfrac{}{}{0pt}{}{\frac{3}{2}}{0,1} \right)}{G^{2,0}_{1,2}\left(\frac{2}{a^2}\middle| \genfrac{}{}{0pt}{}{-\frac{3}{2}}{-3,0} \right)}.
        \end{align}
    \end{subequations}
    More specifically, we are interested in the behavior of \(I(a)\) for small \(a\). To that end, it is useful to use an asymptotic expansion for the Meijer G-functions. For large \(z > 0\), it holds that \cite{luke1969SpecialFunctionsTheirVol1}
    \begin{equation}
        G^{2,0}_{1,2}\left(z\middle| \genfrac{}{}{0pt}{}{\alpha}{\alpha-\beta,0} \right) \sim e^{-z} z^{-\beta} \qty(1 + \sum_{k=1}^{+\infty} M_k z^{-k}),
    \end{equation}
    where \(M_k\) are constants determined by the specific values of \(\alpha\) and \(\beta\). We thus find that 
    \begin{equation}
        \frac{G^{2,0}_{1,2}\left(z\middle| \genfrac{}{}{0pt}{}{\frac{3}{2}}{0,1} \right)}{G^{2,0}_{1,2}\left(z\middle| \genfrac{}{}{0pt}{}{-\frac{3}{2}}{-3,0} \right)} = z + \order{1}.
    \end{equation}
    
    Using this asymptotic expression we conclude that, in the limit with \(a \ll 1\), it holds that
    \begin{equation}\label{eq: Ia-small-a}
        I(a) = a^4 + \order{a^6}.
    \end{equation}
    Hence, \(I(a)\) tends to zero as \(a \to 0\).
\end{appendix}

\bibliography{bibliography}
\end{document}